\documentclass[sigplan,9pt]{acmart}
\renewcommand\footnotetextcopyrightpermission[1]{} %
\copyrightyear{2018}
\acmYear{2018}
\acmConference[LICS '18]{33rd Annual ACM/IEEE Symposium
on Logic in Computer Science}{July 9--12, 2018}{Oxford, United
Kingdom}
\acmBooktitle{Proceedings of 33rd Annual ACM/IEEE Symposium on
Logic in Computer Science (LICS), July 9--12, 2018, Oxford, United Kingdom}
\acmPrice{15.00}
\acmDOI{10.1145/3209108.3209169}
\acmISBN{978-1-4503-5583-4/18/07}
\setcopyright{rightsretained}
\copyrightyear{2018}           %

\usepackage{amsfonts}
\usepackage{galois}
\usepackage[backgroundcolor=orange!50!white,disable]{todonotes}
\presetkeys{todonotes}{inline}{}
\usepackage{listings}
\definecolor{dkblue}{rgb}{0,0.1,0.6}
\definecolor{dkgreen}{rgb}{0,0.35,0}
\definecolor{dkviolet}{rgb}{0.3,0,0.5}
\definecolor{dkred}{rgb}{0.5,0,0}
\lstdefinelanguage{NT}{ 
mathescape=true,
texcl=false, 
morekeywords=[1]{while, do, if, then, else, for, all},
morekeywords=[2]{return, continue, goto},
morecomment=[s]{/*}{*/},
showstringspaces=false,
morestring=[b]",
morestring=[d],
tabsize=3,
extendedchars=false,
sensitive=true,
breaklines=false,
basicstyle=\ttfamily,
captionpos=b,
columns=[l]fixed,
identifierstyle={\color{black}},
keywordstyle=[1]{\color{dkviolet}},
keywordstyle=[2]{\color{dkred}},
stringstyle=\ttfamily,
commentstyle={\ttfamily\color{dkblue}},
literate=
    {true}{{{\color{dkgreen}$true$}}}3
    {false}{{{\color{dkgreen}$false$}}}4
}[keywords,comments,strings]
\lstnewenvironment{codeNT}{\lstset{language=NT}}{}
\usepackage[curve,matrix,arrow,cmtip]{xy}
\usepackage{paralist}

\usepackage{stmaryrd} %
\newcommand{\bb}[1]{\llbracket#1\rrbracket}
\newcommand{\Rel}{\mathsf{Rel}}
\newcommand\dfa[1]{{%
  \newcommand\sstate[1]{##1}%
  \newcommand\fstate\overline%
  #1}}
\newcommand{\Eqv}{e}
\newcommand{\Pred}{\mathsf{Pred}}
\newcommand{\Sign}{\mathsf{Sign}}
\newcommand{\Nat}{\mathbb{N}}
\newcommand{\wrt}{w.r.t.}

\AtEndPreamble{%
\theoremstyle{acmdefinition}

\newtheorem{remark}[theorem]{Remark}
\newtheorem{assumption}[theorem]{Assumption}
}

\begin{document}
\title{Sound up-to techniques and Complete abstract domains}         %
\author{Filippo Bonchi}
\affiliation{
  \institution{University of Pisa}            %
  \city{Pisa}
  \country{Italy}                    %
}
\email{filippo.bonchi@unipi.it}          %
\author{Pierre Ganty}
\affiliation{
  \institution{IMDEA Software Institute}           %
	\city{Pozuelo de Alarc{\'o}n}
  \country{Spain}                   %
}
\email{pierre.ganty@imdea.org}
\author{Roberto Giacobazzi}
\affiliation{
  \institution{IMDEA Software Institute}           %
	\city{Pozuelo de Alarc{\'o}n}
  \country{Spain}                   %
}
\affiliation{
  \institution{and University of Verona}            %
  \city{Verona}
  \country{Italy}                    %
}
\email{roberto.giacobazzi@imdea.org}         %
\authornote{Partially supported by AFOSR.}          %
\author{Dusko Pavlovic}
\authornote{Partially supported by AFOSR.}          %
\affiliation{
  \institution{University of Hawaii}            %
  \city{Honolulu}
  \state{Hawaii}
  \country{USA}                    %
}
\email{dusko@hawaii.edu}          %

\begin{abstract}
Abstract interpretation is a method to automatically find invariants of programs or pieces of code whose semantics is given via least fixed-points.
Up-to techniques have been introduced as enhancements of coinduction, an abstract principle to prove properties expressed via greatest fixed-points. 

While abstract interpretation is always sound by definition, the soundness of up-to techniques needs some ingenuity to be proven. For completeness, the setting is switched: up-to techniques are always complete, while abstract domains are not. 

In this work we show that, under reasonable assumptions, there is an evident connection between sound up-to techniques and complete abstract domains.
\end{abstract}

\maketitle
\renewcommand{\shortauthors}{F. Bonchi, P. Ganty, R. Giacobazzi, D. Pavlovic}

\section{Introduction}
Abstract interpretation \cite{cousot1977abstract} is a general method for approximating invariants of dynamic systems. The key idea is that the analysis, or the possibility of proving invariance properties of a system, can be reduced to compute an approximate semantics of the system under inspection. Any monotone function $b$ on a monotone lattice $C$ can be approximated in a sound way by a function $\overline{b}$ on an smaller lattice $A$ and an approximate invariant can be always obtained by computing the least fixed point of $\overline{b}$. 
Instances of abstract interpretation include sound-by-construction methodologies for the design of static program analyses, type analysis~\cite{Cousot1997}, and model checkers~\cite{CousotC99}. %
Examples of successful industrialisation cases of abstract interpretation in the context of program analysis are {\sc Astr{\'e}e}  \cite{Esop05Astree},
a static program analyser aiming at proving the absence of run-time errors in industrial-size safety-critical C programs, Julia, for the analysis of Java and Android code for web applications, Clousot, developed at Microsoft Research~\cite{FahndrichL2010} for statically checking Code Contracts in .Net, Infer and Zoncolan, at Facebook Inc.\ for scalable information-flow analysis~\cite{OHearn2015}.

One of the main limitation of abstract interpretation is \emph{completeness}: the least fixed point of the approximate $\overline{b}$ is not always a faithful representation of the one of $b$, hence the latter could satisfy some properties not satisfied by the former. Computing on \(A\) rather than on \(C\) can thus lead to false alarms. Completeness should be intended as absence of false alarms. 

This was first observed by \citeauthor{cousot1979systematic}~\cite{cousot1979systematic}, where they also show a strategy to prove completeness of abstract domains: one can more easily prove a sufficient condition, that we call here \emph{full completeness}, but that is known under difference names, like backward completeness \cite{GiacobazziQ01} or (stepwise) completeness \cite{GiacobazziRS00}. Interestingly enough, \citeauthor{GiacobazziRS00}~\cite{GiacobazziRS00} observed that both completeness and full completeness can be regarded not as a property of $\overline{b}$, but rather of $b$ and $A$. %
A symmetric condition, called \emph{forward completeness}, was later introduced in \cite{GiacobazziQ01} and used  for an efficient simulation equivalence algorithm \cite{RanzatoTLics07}.
The key insight here is that when $b$ is a left adjoint, then full completeness coincides with forward completeness of its right adjoint.

\smallskip

The rationale behind coinductive up-to techniques is apparently dual.
Suppose we have a characterisation of an object of interest as a
greatest fixed-point of some function. For instance, behavioural equivalence in CCS~\cite{Milner89} is
the greatest fixed-point of a monotone function $b$ on the lattice of relations,
describing the standard bisimulation game. This means that to prove
two CCS terms equivalent, it suffices to exhibit a relation $R$ that
relates them, and which is a \emph{$b$-simulation}, i.e., $R\subseteq
b(R)$. 

Alternatively, one can look for a relation \(R\) which is a $b$-simulation \emph{up to some function} $a$, i.e., $R\subseteq b(a(R))$.
However, not every function $a$ can safely be used: $a$ should be \emph{sound} for $b$, meaning that any $b$-simulation up to $a$ should be contained in a $b$-simulation. 
A similar phenomemon occurs in abstract interpretation where not all abstract domains are complete.

Since their introduction~\cite{Milner89},
coinduction up-to techniques were proved useful, if not essential, in numerous proofs
about concurrent systems (see~\cite{PS12} for a list of
references); it has been used to obtain decidability
results~\cite{Caucal90}, and more recently to improve standard
automata algorithms~\cite{bp:popl13:hkc}. It is worth to make clear at this point that, while abstract interpretation was originally  intended as a fully automated approach to program analysis, coinduction up-to has always been seen as a proof principle: this explains the increasing spread of up to techniques amongst proof assistants (see e.g.~\cite{danielsson2017up}).

Since proving soundness of these techniques is rather complicated and error prone -- a famous example of an unsound technique is that of weak bisimulation up to weak bisimilarity -- \citeauthor{San98MFCS} introduced the sufficient condition of respectfulness~\cite{San98MFCS}. This was later refined by \citeauthor{pous:aplas07:clut} with the notion of \emph{compatibility}~\cite{pous:aplas07:clut,PS12}. 

\medskip

In this paper we relate abstract interpretation with coinduction up-to. 
Our key observation (Remark~\ref{rmk:fwcompletenesscompatibility}) is that the notions of forward completeness and compatibility coincide. 
Using the key insight of \citeauthor{GiacobazziRS00}~\cite{GiacobazziRS00} saying that when $b$ is a left adjoint,  full completeness of $b$ coincides with compatibility of its right adjoint, we prove that the same abstraction can play the role of a complete abstract domain for abstract interpretation and a sound up-to technique for coinduction. 

As a noteworthy example of this correspondence, we show in this paper an abstract-interpretation based analysis of the well-known Hopcroft and Karp's algorithm~\cite{HopcroftKarp,Aho:1974:DAC:578775} for checking language equivalence of deterministic automata. In this case the optimisation allowing the reduction of the state space, which is known from up-to techniques can also be derived by abstract interpretation.

Finally the connection between complete abstract interpretations and sound up-to techniques allows some technology transfer.
As a proof of concept, we introduce in abstract interpretation the notion of companion, already known in coinduction up-to, that provides a way to simplify the checking of completeness for a generic abstraction. This also leads to the definition of a new and weaker notion of completeness for abstract interpretation which is {\em local\/} but sufficient to prove the absence of false alarms in program analysis. 

\section{Preliminaries and notation}\label{sec:preliminaries}
We use $(L,\sqsubseteq)$, $(L_1,\sqsubseteq_1)$, $(L_2,\sqsubseteq_2)$ to range over complete lattices and $x,y,z$ to range over their elements. We omit the ordering $\sqsubseteq$ whenever unnecessary. As usual $\bigsqcup$ and $\bigsqcap$ denote  least upper bound and greatest lower bound, $\sqcup$ and $\sqcap$ denote join and meet, $\top$ and $\bot$ top and bottom.

Hereafter we always consider monotone maps so we will often omit to specify that they are monotone. 
Monotone maps form a complete lattice with their natural point-wise order: whenever $f,g\colon L_1\to L_2$ then 
$f\sqsubseteq g$ iff for any $x\in L_1$: $f(x)\sqsubseteq_2 g(x)$. 
Obviously, the identity  $id\colon L\to L$ and the composition $f\comp g \colon L_1\to L_3$ of two monotone maps $g\colon L_1\to L_2$ and  $f\colon L_2\to L_3$  are monotone.

Given a monotone map $f\colon L\to L$, $x\in L$ is said to be a \emph{post-fixed point}\footnote{It is also common to find in literature the reversed definitions of post and pre-fixed point. Here we adopted the terminology of \citeauthor{davey_priestley_2002}~\cite{davey_priestley_2002}.} if{}f  $x\sqsubseteq f(x)$ and a \emph{pre-fixed point} if{}f $f(x)\sqsubseteq x$. A \emph{fixed point} if{}f $x=f(x)$. 
Pre, post and fixed points form complete lattices, denoted by $Pre(f)$, $Post(f)$ and $Fix(f)$, respectively. 
We write $\mu f$ and $\nu f$ for the least and greatest fixed-point.

For a map $f\colon L \to L$, we inductively define $f^0=id$ and $f^{n+1}=f\comp f^n$. We fix $f^\uparrow = \bigsqcup_{i\in \Nat} f^i$ and $f^\downarrow = \bigsqcap_{i\in \Nat} f^i$. A monotone map $f\colon L \to L$ is an \emph{up-closure} operator if $x\sqsubseteq f(x)$ and $ff(x) \sqsubseteq f(x)$. It is a \emph{down-closure} operator if $f(x)\sqsubseteq x$ and $f(x) \sqsubseteq ff(x)$. For any $f$, $f^\uparrow$ is an up-closure and $f^\downarrow$ is a down-closure~\cite{CC79b}. 

Given $l \colon L_1 \to L_2$ and $r\colon L_2\to L_1$, we say that $l$ is the \emph{left adjoint} of $r$, or equivalently that $r$ is the \emph{right adjoint} of $l$, written \((L_1,\sqsubseteq_1) \galois{l}{r} (L_2,\sqsubseteq_2)\), exactly when we have $l(x)\sqsubseteq_2 y$ if{}f $x \sqsubseteq_1 r(y)$ for all $x\in L_1$ and $y\in L_2$. 
Moreover, $r\comp l$ is an up-closure operator and $l\comp r$ a down-closure operator. 
When $l\comp r=id$ we say that $l,r$ form a \emph{Galois insertion} (GI), hereafter denoted as $(L_1,\sqsubseteq_1) \galoiS{l}{r} (L_2,\sqsubseteq_2)$.
Closure operators and Galois insertions are bijective  correspondence: given an up-closure operator $f\colon L\to L$, the functions $l\colon L \to Pre(f)$, defined as $l(x) = \bigsqcap \{y \mid x\sqsubseteq y \sqsupseteq f(y) \}$ is the left adjoint of $r\colon Pre(f) \to L$, defined as $r(x)=x$.

\medskip

For a monotone map $b\colon L \to L$ on a complete lattice $L$, the Knaster-Tarski fixed-point theorem characterises $\mu b$ as the least upper bound of all pre-fixed points of $b$ and $\nu b$ as the
greatest lower bound of all its post-fixed points:
\begin{equation*}\label{eq:KNfpthm}
\mu b= \bigsqcap \{ x  \mid b(x) \sqsubseteq x \} \qquad \qquad \nu b= \bigsqcup \{ x  \mid x \sqsubseteq b(x) \}\enspace .
\end{equation*} 
This immediately leads to the
\emph{induction} and \emph{coinduction} proof principles, illustrated below, on the left and on the right, respectively
\cite{Park69}.
\begin{equation}\label{eq:coinductionproofprinciple}\begin{array}{c}
    \exists x, \;  b(x)\sqsubseteq x \sqsubseteq f\\
    \hline 
    \mu b \sqsubseteq f
\end{array}
\qquad
\qquad
\begin{array}{c}
    \exists x, \; i \sqsubseteq x\sqsubseteq b(x)\\
    \hline 
    i \sqsubseteq \nu b
\end{array}
\end{equation}

Another fixed-point theorem, usually attributed to Kleene, plays an important role in our exposition. It characterises $\mu b$ and $\nu b$ as the least upper bound, respectively the greatest lower bound, of the chains
\begin{equation}\label{eq:initfinsequences}
\bot \sqsubseteq b(\bot) \sqsubseteq bb(\bot) \sqsubseteq  \dots \qquad  \top \sqsupseteq b(\top) \sqsupseteq bb(\top)  \sqsupseteq  \dots 
\end{equation}
In short, 
\begin{equation*}\label{eq:Kleenefpthm}
\mu b = \bigsqcup_{i\in \Nat} b^i(\bot) \qquad \qquad \nu b = \bigsqcap_{i\in \Nat} b^i(\top)
\end{equation*}
The assumptions are stronger than for Knaster-Tarski: for the leftmost statement, it requires the map $b$ to be \emph{Scott-continuous} (i.e.., it preserves $\bigsqcup$ of directed chains) and, for the rightmost  \emph{Scott-cocontinuous} (similar but for  $\bigsqcap$). Observe that every left adjoint is continuous (it preserves aribtrary $\bigsqcup$) and every right adjoint is cocontinuous.

\medskip

Coinduction up-to can be thought as an optimisation of the principle in \eqref{eq:coinductionproofprinciple}, right. Abstract interpretation as an optimisation of the chain in \eqref{eq:initfinsequences}, left. In both cases, the optimisation is given by an up-closure.

\section{Coinduction up-to} \label{sec:upto}
In order to motivate up-to techniques we illustrate how coinduction can be exploited to check language equivalence of automata.

\subsection{Coinduction for Deterministic Automata}\label{sec:DA}
A deterministic automaton on the alphabet $A$ is a triple $(X,
o,t)$, where $X$ is a set of states, $o\colon X\to 2=\{0,1\}$ is the output function, determining if a state $x$ is final ($o(x) = 1$) or
not ($o(x) = 0$) and $t
\colon X \to X^A$ is the transition function which returns  the next state, for each
letter $a \in A$.

Every automaton $(X, o,t)$ induces a function
$\bb{-}\colon X \to 2^{A^*}$  defined for all $x\in
X$, $a \in A$ and $w\in A^*$ as 
$
\bb{x}(\varepsilon) =  o(x)$ and 
$\bb{x}(aw)       =    \bb{t(x)(a)}(w)
$.
Two states $x,y\in X$ are said to be \emph{language equivalent}, in symbols
$x \sim y$, if{}f  $\bb{x}=\bb{y}$.
Alternatively, language equivalence can be defined coinductively as the greatest
fixed-point of a map $b$ on $\Rel_X$, the lattice of relations
over $X$.  For all $R\subseteq X^2$, $b\colon \Rel_X \to \Rel_X$ is
defined as
\begin{multline}\label{eq:functional-bisim-da}
	b(R)=\{(x,y) \mid o(x)=o(y) \text{ and,}\\
		\text{for all } a\in A, (t(x)(a), t(y)(a))\in R  \}\enspace .
\end{multline}
Indeed, one can check that $b$ is monotone and that $\nu b = \mathord{\sim}$. 

Thanks to this characterisation, one can prove $x\sim y$ by mean of the coinduction proof principle illustrated in~\eqref{eq:coinductionproofprinciple}. 
To this end, one provides a relation $R$ that is a \emph{$b$-simulation}:  a post fixed-point of $b$. 
Besides being a \(b\)-simulation, \(R\) must satisfy $\{(x,y)\}\subseteq R$.

For an example, consider the following deterministic
automaton, where final states are over lined and the transition
function is represented by labeled arrows. The relation consisting of
dashed and dotted lines is a $b$-simulation showing that $x\sim u$.
\begin{equation}\label{eq:exautomata}
	\begin{minipage}[c]{\columnwidth}
			\dfa{\xymatrix @R=.5em@C=1.5cm { 
			    \sstate{x}\ar[r]^{a,b}\ar@{--}[dddd]& 
			    \fstate{y}\ar[r]^{a,b}\ar@{--}[ddd] \ar@{.}[rddd] &
			    \fstate{z}\ar@(ur,dr)^{a,b}\ar@{--}[ddd] \ar@{--}[lddd]\\\\\\
			    & \fstate{v}\ar@/^/[r]^{a,b}& 
			    \fstate{w}\ar[l]^{a,b}\\
			    \sstate{u}\ar[ru]^a\ar@/_1.1em/[rru]^b& }}
	\end{minipage}
\end{equation}

Figure~\ref{fig:naive} illustrates an algorithm, called \texttt{Naive}, that takes in input a deterministic automaton $(X, o,t )$ and a pair of states $(x_1,x_2)$. It attempts to build a bisimulation $R$ containing $(x_1,x_2)$: if it succeeds, then $x_1\sim x_2$ and  returns true, otherwise returns false.
\begin{figure}[t]
\centering
\underline{\texttt{Naive} $(x_1,x_2)$}
\begin{codeNT}
(1) $R := \emptyset$; $todo := \emptyset$
(2) insert $(x_1,x_2)$ into $todo$
(3) while $todo$ is not empty do 
   (3.1)  extract $(x_1',x_2')$ from $todo$
   (3.2)  if $(x_1',x_2')\in R$ then goto (3)
   (3.3)  if $o(x_1')\neq o(x_2')$ then return false
   (3.4)  for all $a\in A$, 
	           insert $(t(x_1')(a),\,t(x_2')(a))$ into $todo$
   (3.5)  insert $(x_1',x_2')$ into $R$ 
(4) return true
\end{codeNT}
\caption{Naive algorithm checking language equivalence of states $x_1,x_2\in X$ for a deterministic automaton $(X, o,t)$.}
\label{fig:naive}
\end{figure}

The worst case complexity of the algorithm \texttt{Naive} is linear in the size of the computed bisimulation $R$. Therefore, it is quadratic with respect to the number of states in $X$. An optimised version of \texttt{Naive}, the well known Hopcroft and Karp algorithm \cite{Aho:1974:DAC:578775,HopcroftKarp}, can be given by means of  up-to techniques.

\subsection{Up-to techniques}
Coinduction allows to prove $i\sqsubseteq \nu b$ for a given map $b \colon C \to C$ on a complete lattice $C$ and some $i\in C$. Up-to techniques have been introduced by \citeauthor{Milner89}~\cite{Milner89} as an enhancement for coinduction. In a nutshell, an \emph{up-to technique} is a monotone map $a \colon C \to C$. A \emph{$b$-simulation up to $a$} is a post-fixed point of $ba$, that is an $x$ such that $x\sqsubseteq ba(x)$. An up-to technique $a$ is said to be \emph{sound} \wrt{} $b$ (or \(b\)-sound, for short) if the following \emph{coinduction up to} principle holds.

\begin{equation}\label{eq:coinductionuptoproofprinciple}
 \begin{array}{c}
    \exists x, \; i \sqsubseteq x\sqsubseteq ba(x)\\
    \hline 
    i \sqsubseteq \nu b
\end{array}
\end{equation}
An equivalent formulation can be given as follows.  
\begin{lemma}\label{lemma:charsound}
$a$ is $b$-sound if{}f $\nu ba \sqsubseteq \nu b$. 
\end{lemma}

\begin{remark}[Completeness of up-to technique]\label{rmk:completenessupto}

\noindent
Observe that, according to the above definition an up-to technique $a$ might not be \emph{complete}: it may exist an $i$ such that $i \sqsubseteq \nu b$ for which there is no $x$ satisfying  $i \sqsubseteq x\sqsubseteq ba(x)$. However, if $a$ is an up-closure operator, then $\nu b  \sqsubseteq a(\nu b)$ and using monotonicity of $b$, one obtains that 
$i \sqsubseteq \nu b = b(\nu b) \sqsubseteq b(a (\nu b))$. The question of completeness for up-to techniques has never been raised because they have always been considered up-closure operators (e.g., up-to equivalence, up-to congruence). The main reason for considering arbitrary monotone maps rather than just up-closure operators, comes from the fact that the former allows for more modular proofs of their soundness. This is discussed in more details in Remark~\ref{rmk:mod}.
\end{remark}

\subsection{Hopcroft and Karp's algorithm from the standpoint of up-to techniques}\label{ssec:HK}
For an example of  up-to technique, take the function
$\Eqv\colon \Rel_X \to \Rel_X$ mapping every relation $R\subseteq X^2$
to its equivalence closure. We will see in Section~\ref{sec:HKsoundness}, that $\Eqv$ is sound for the map $b$ defined in~\eqref{eq:functional-bisim-da}. A \emph{$b$-simulation up to $\Eqv$} is a
relation $R$ such that $R\subseteq b (\Eqv (R))$.
Consider the automaton in~\eqref{eq:exautomata} and the relation $R$
containing only the dashed lines: since $t(x)(b)=y$, $t(u)(b)=w$ and
$(y,w)\notin R$, then $(x,u)\notin b(R)$. This means that $R$ is
\emph{not} a $b$-simulation; however it is a $b$-simulation up to $\Eqv$,
since $(y,w)$ belongs to $\Eqv (R)$ and $(x,u)$ to $b(\Eqv(R))$.

This example shows that $b$-simulations up-to $\Eqv$ can be smaller than plain $b$-simulations: this idea is implicitly exploited in the Hopcroft and Karp's algorithm~\cite{Aho:1974:DAC:578775,HopcroftKarp} to check language equivalence of deterministic automata. This algorithm can be thought as an optimisation of \texttt{Naive}, where line \texttt{(3.2)} is replaced with the following\footnote{In Section \ref{sec:standpoint}, we will se that there is a little, but significant, difference between this and the original formulation from~\cite{Aho:1974:DAC:578775}.}.
\begin{codeNT}
	(3.2)  if $(x_1',x_2')\in e(R)$ then goto (3)
\end{codeNT}
This optimised algorithm skips any pair which is in the equivalence closure of $R$: during the while loop \texttt{(3)}, it always holds that
$R\subseteq b(e(R) \cup todo)$. The algorithm returns true only if $R\subseteq be(R)$. This means that $R$ is a $b$-simulation up to $e$ containing $(x_1,x_2)$.

This simple optimisation allows to reduce the worst case complexity of \texttt{Naive}: the size of the returned relation $R$ cannot be larger than $n$ (the number of states).
The case of non-deterministic automata is even more impressive:
another up-to technique, called \emph{up-to congruence}, allows for
an exponential improvement~\cite{bp:popl13:hkc}. 

\begin{remark}\label{rmk:hopcroft}
	The partition refinement algorithm by \citeauthor{hopcroft1971n}~\cite{hopcroft1971n} computes language equivalence for deterministic automata by constructing the chain defined in the right of~\eqref{eq:initfinsequences} for the $b$ in~\eqref{eq:functional-bisim-da},  e.g., $\top \sqsupseteq \{x,u\}\{y,v,w,z\}\sqsupseteq \{x,u\}\{y,v,w,z\}$ for the automaton in~\eqref{eq:exautomata}. 

The crucial observation, for showing that this chain stabilises after at-most $n$ iterations is that every element of the chain is an equivalence relation. Somehow, the computation of the chain for the greatest fixed point is already up-to equivalence. This fact will find a deeper explanation at the end of Section~\ref{sec:duality}.
\end{remark}

\section{Abstract Interpretation}\label{sec:ai}
We introduce abstract interpretation by showing a simple problem of program analysis.

\subsection{A toy program analysis}\label{sec:toyprogramanalysis}
Consider the following piece of code, where \texttt{$x$} is an integer value.
\begin{codeNT}
  $x := 5$;  while $x>0$ do { $x:=x-1$; }
\end{codeNT}
We want to prove that after exiting the loop, \texttt{$x$} has value $0$. Our analysis works on the lattice of predicates over the integers, hereafter denoted by $Pred_Z$, and makes use of the function $\ominus \!1 \colon Pred_Z \to Pred_Z$ defined as $P\ominus \!1 = \{i-1\mid i\in P  \}$, for all $P\in \Pred_Z$. We start by annotating the code so to make explicit its control flow.
\begin{codeNT}
	$x := 5$;$^{1}$  while ${}^{2} x>0 {}^{3}$  do { $x:=x-1$;$^{4}$ }$^{5}$
\end{codeNT}
We then write the following system of equations where $x^j$ contains the set of possible values that  \texttt{$x$} can have at the position $j$.
\begin{equation*}
	x^1{=}\{5\},  \; x^2 {=} x^1\cup x^4, \; x^3 {=} x^2 \cap [1,\infty), \; x^4 {=} x^3 \ominus \!1, \;  x^5 {=} x^2 \cap ({-}\infty, 0]
\end{equation*}
For $x^2$, we obtain the equation
$ x^2 = \{5\} \cup ( (x^2 \cap [1,\infty) ) \ominus \!1 )$ that has as smallest solution $\mu b$ where $b\colon Pred_Z \to Pred_Z$ is defined as 
\begin{equation}\label{eq:babstractInt}
b(P) = \{5\} \cup ( (P\cap [1,\infty) ) \ominus\! 1 )
\end{equation} for all predicates $P$. Our initial aim is to check whether $x^5 = x^2 \cap (-\infty, 0] = \mu (b) \cap (-\infty, 0]  \subseteq \{0\}$. That is $\mu b \subseteq [0,\infty)$. 

We proceed by computing $\mu b$ as in the left of~\eqref{eq:initfinsequences}:
\begin{equation}\label{eq:lfpcomp}
\emptyset \subseteq \{5\} \subseteq \{5,4\} \subseteq \dots \subseteq \{5,4,3,2,1\} \subseteq \{5,4,3,2,1,0\}
\end{equation}
Since $\{5,4,3,2,1,0\} \subseteq [0,\infty)$, we have proved our conjecture.

\subsection{Abstract domains}
We consider the standard Galois insertion-based definition of an abstract domain \cite{cousot1977abstract}.
Define an \emph{abstract domain} to be a Galois insertion $C \galoiS{\alpha}{\gamma} A$ between complete lattices \(C\) and \(A\).
Sometimes we call an abstract domain the associated up-closure, hereafter denoted by $a\colon C \to C$ \cite{cousot1979systematic}. We will always identify $A$ with $Pre(a)$.
The main idea of abstract interpretation is that in order to check whether $\mu b \sqsubseteq f$, for a map $b\colon C \to C$ and $f\in C$, the computation of $\mu b$ via the chain in~\eqref{eq:initfinsequences} can be carried more efficiently in some abstract domain $A$. One wants to define some $\overline{b}\colon A\to A$ representing an approximation of $b$ in $A$ and then check whether $\mu \overline{b} \sqsubseteq f$. Note that, in the latter inequality, the left hand side  stands in $A$, while the right one in $C$. For this reason, it is always assumed that $f\in A=Pre(a)$, that is $a(f)\sqsubseteq f$.

An approximation $\overline{b}$ is said to be \emph{sound} (\wrt{} $b$) if
$\alpha b \sqsubseteq \overline{b} \alpha$.\footnote{Soundness is often defined also by the equivalent inequation $\alpha b \gamma \sqsubseteq \overline{b}$.} 
The terminology is justified by the last point of the following lemma by Cousot and Cousot \cite{cousot1979systematic}.
\begin{lemma}\label{lemma:absouncomplete}
	Let $b,a\colon C \to C$ be a map and a closure operator with associated Galois insertion \(C \galoiS{\alpha}{\gamma} C\). Let $f\in Pre(a)$.
\begin{enumerate}
\item Let $x\in C$. Then  $x \sqsubseteq_C f$ if{}f $\alpha(x) \sqsubseteq_A f$ if{}f $a(x) \sqsubseteq_C f$.
\item In particular, $\mu b \sqsubseteq_C f$ if{}f $\alpha(\mu b) \sqsubseteq_A f$ if{}f $a(\mu b) \sqsubseteq_C f$.
\item  If $\overline{b}$ is $b$-sound, then $\alpha(\mu b) \sqsubseteq_A \mu \overline{b}$.
\item If $\overline{b}$ is $b$-sound, then $\mu \overline{b} \sqsubseteq_A f$ entails that $\mu b \sqsubseteq_C f$.
\end{enumerate} 
\end{lemma}
For a monotone map $b$ and an up-closure $a$, we define $\overline{b}^a$ as $\alpha \comp b \comp \gamma $. As explained by the following proposition, this approximation plays a key role.

\begin{proposition}[\cite{cousot1979systematic}, Corollary 7.2.0.4]\label{prop:alwayssound}
The map $\overline{b}^a$ is \emph{the best sound approximation}, that is:
\begin{inparaenum}[\upshape(\itshape 1\upshape)]
	\item $\overline{b}^a$ is sound and
	\item if $\overline{b}$ is sound, then  $\overline{b}^a \sqsubseteq \overline{b}$.
\end{inparaenum}
\end{proposition}

Therefore, for all abstract domain $a$, there exists a sound approximation $\overline{b}^a$, that is $\mu \overline{b}^a \sqsubseteq_A f$ implies that $\mu b \sqsubseteq_C f$.

The converse implication is not guaranteed in general. One has to require completeness of the abstract domain: $a$ is \emph{complete} \wrt{}  $b$ (or $b$-complete, for short) if{}f $\alpha(\mu b) = \mu  \overline{b}^a$
\cite{cousot1979systematic,GiacobazziRS00}. Standard completeness in abstract interpretation is called here $b$-completeness. 

\begin{lemma}\label{lemma:characterizationcompleteness}
Let $b,a\colon C \to C$ be a map and a closure operator. 

If $a$ is $b$-complete then for all $f\in Pre(a)$, $\mu \overline{b}^a \sqsubseteq_A f$ if{}f $\mu b \sqsubseteq_C f$.
\end{lemma}

We will often find convenient the following alternative characterization.
\begin{lemma}[\cite{GiacobazziRS00}, Lemma 3.1]\label{lemma:alternativecompleteness}
$a$ is $b$-complete if{}f $a(\mu b) = \mu (a b) $.
\end{lemma}

\subsection{Abstract Interpretation of a toy program}\label{ssec:AItoy}
Consider $\Sign_Z$, the abstract domain of signs depicted below: each element is a predicate over $Z$. The right adjoint $\gamma\colon \Sign_Z \to \Pred_Z$ is the obvious inclusion and the left adjoint $\alpha \colon \Pred_Z \to \Sign_Z$ maps any predicate $P$ into the smallest $Q$ in  $\Sign_Z$, such that $P\subseteq Q$. For instance, $\alpha(\{5,6\})=[1,\infty)$. Take $s\colon \Pred_Z \to \Pred_Z$ as $\gamma \comp \alpha$. As recalled in Section~\ref{sec:preliminaries}, $Pre(s)=\Sign_Z$.

\[{\lower-1.2cm\hbox{\xymatrix@R=0.3cm@C=0.6cm{ & Z &\\
	(-\infty,0] \ar@{-}[ur]  & (-\infty,-1] \cup [1,\infty) \ar@{-}[u]   &[0,\infty) \ar@{-}[ul] \\
	(-\infty,-1]  \ar@{-}[u] \ar@{-}[ur]  & \{0\} \ar@{-}[ur]\ar@{-}[ul]& [1,\infty) \ar@{-}[ul]  \ar@{-}[u]\\
	& \emptyset \ar@{-}[ur] \ar@{-}[ul] \ar@{-}[u]}
	}}\]

	For $b$ defined as in~\eqref{eq:babstractInt}, its best sound approximation is $\overline{b}^s(P)=   [1,\infty) \sqcup ((P\sqcap [1,\infty))\overline{\ominus\!1}^s)$ where $\overline{\ominus \!1}^s$ is again defined as $\alpha \comp \ominus \!1 \comp \gamma$, e.g.,  $\overline{\ominus \!1}^s([1,\infty))= [0,\infty)$. The computation of $\mu \overline{b}^s$ is shorter than the one of $\mu b$ in~\eqref{eq:lfpcomp}:
$\emptyset \sqsubseteq [1,\infty) \;\sqsubseteq  [0,\infty) \sqsubseteq [0,\infty)$. Since $\mu \overline{b}^s \subseteq [0,\infty)$ and since $\overline{b}^s$ is sound, one can conclude that $\mu b \subseteq [0,\infty)$.

\medskip

Imagine now that one would like to check whether, after the while loop in the toy program, \texttt{$x$} has a negative value. It is necessary to verify that  $\mu b \cap (-\infty, 0]  \subseteq (-\infty,-1]$, i.e., $\mu b \subseteq (-\infty,-1]$. After computing $\mu \overline{b}^s=[0,\infty) \not \sqsubseteq (-\infty,-1]$, one would like to conclude that the property does not hold. This deduction cannot be done in general, but only when the abstract domain is complete. %

In this case, it is pretty easy to see that $\Sign_Z$ is complete: use Lemma~\ref{lemma:alternativecompleteness} and observe that $s(\mu b)= s(\{5,4,3,2,1,0\})=[0,\infty) = \mu(sb)$. However, without knowing the value of $\mu b$, proving the completeness is rather complicated. For this reason, in the next section, we will illustrate a sufficient condition entailing completeness.

\section{Proving soundness and completeness}\label{sec:soundandcomplete}
Sufficient conditions were introduced to prove soundness of up-to techniques and completeness of abstract domains. Next, we report on several equivalent formulations of such conditions.

\begin{lemma}\label{lemma:equivalentformulation}
	Let $b\colon C\to C$ be a monotone map and $a\colon C\to C$ an up-closure operator with $C\galoiS{\alpha}{\gamma}A$ as associated pair of adjoint maps, i.e., $Pre(a)=A$. The followings are equivalent:
\begin{enumerate}[{\textbullet{}}1]
\item $ba = aba$;
\item $ab \sqsubseteq ba$ (EM law);
\item there exists a  $\overline{b}\colon A \to A$ such that  $\gamma \overline{b} =  b \gamma$ (EM lifting).
\end{enumerate}
The followings are equivalent:
\begin{enumerate}[{\(\circ\)}1]
\item $ab = aba$;
\item $ba \sqsubseteq ab$ (Kl law);
\item there exists a $\overline{b}\colon A \to A$ such that $\overline{b} \alpha = \alpha b$ (Kl extension).
\end{enumerate}
\end{lemma}
These facts are well known and appear in different places in literature: the reader can find a proof in
Appendix~\ref{app:proof}.
We call a monotone map $a\colon C \to C$  \emph{compatible} \wrt{}  $b$ if it enjoys the property $\bullet 2$ in Lemma~\ref{lemma:equivalentformulation}; we call it \emph{fully complete} \wrt{}  $b$, if it enjoys $\circ 2$. Compatibility entails soundness for up-to techniques, while full completeness entails completeness for abstract domains.
\begin{theorem}[\cite{PS12}, Theorem 3.6.9]\label{prop:compatible}
Let $b,a\colon C \to C$ be a map and an up-closure. If $a$ is compatible \wrt{}  $b$, then \(a\) is sound \wrt{}  $b$.
\end{theorem}
\begin{theorem}[\cite{cousot1979systematic}, Theorem 7.1.0.4]\label{prop:Cousot} 
Let $b,a\colon C \to C$ be a map and an up-closure. If $a$ is fully complete \wrt{}  $b$, then \(a\) is complete \wrt{}  $b$. 
\end{theorem}
It is important to remark here that both theorems state sufficient conditions that are not, in general, necessary: we have seen in Section~\ref{ssec:AItoy} that the abstract domain of signs is complete but, as we will see in Section~\ref{sec:Signcomplete}, it is not  fully complete.
Indeed, compatibility and full completeness, as characterised by points $\bullet 3$ and $\circ 3$ of Lemma~\ref{lemma:equivalentformulation}, require step wise correspondences between $b$ and $\overline{b}$, visualized below,
\[
\xymatrix@R=0.2cm@C=0.2cm{
A \ar[rr]^{\overline{b}} \ar[dd]_\gamma & & A  \ar[dd]^\gamma \\ 
&  \\
C \ar[rr]_b & &  C 
}
\qquad
\qquad
\xymatrix@R=0.2cm@C=0.2cm{A \ar[rr]^{\overline{b}} & & A\\ 
&  \\
C \ar[rr]_b \ar[uu]^\alpha & &  C \ar[uu]_\alpha
}
\]
while soundness and completeness require the correspondences just of the fixed points (that can happen to hold for different reasons).

The formulations of compatibility and full completeness provided by points $\bullet 2$ and $\circ 2$ of Lemma~\ref{lemma:equivalentformulation} are more handy to make modular proofs of compatibility and full-completeness, as shown in the next proposition. For compatibility, $h$, ($h_1$, $h_2$) in Proposition~\ref{prop:mod} will play the role of $b$, while $g$, ($g_1$, $g_2$) the role of the up-to technique $a$. For abstract domains instead, the situation is reversed: $h$, ($h_1$, $h_2$) will play the role of $a$, while $g$, ($g_1$, $g_2$) the role of $b$.

\begin{proposition}[modularity]\label{prop:mod}
Let $g,h,g_1,g_2,h_1,h_2\colon C\to C$ be monotone maps on some complete lattice $C$. 
Then:
\begin{enumerate}[1.]
\item $id \comp h \sqsubseteq h \comp id$;
\item if  $g_1 \comp h \sqsubseteq h \comp g_1$ and $g_2 \comp h \sqsubseteq h \comp g_2$, then $(g_1 \comp g_2) \comp h \sqsubseteq h \comp (g_1 \comp g_2)$.
\end{enumerate}
Moreover:
\begin{enumerate}[1.]
\setcounter{enumi}{2}
\item if  $g_1 \comp h \sqsubseteq h \comp g_1$ and $g_2 \comp h \sqsubseteq h \comp g_2$, then $(g_1 \sqcup g_2) \comp h \sqsubseteq h \comp (g_1 \sqcup g_2)$;
\item if  $g \comp h \sqsubseteq h \comp g$, then $g^\uparrow \comp h \sqsubseteq h \comp g^\uparrow$.
\end{enumerate}
Dually:
\begin{enumerate}[1.]
\setcounter{enumi}{4}
\item if  $g \comp h_1 \sqsubseteq h_1 \comp g$ and $g \comp h_2 \sqsubseteq h_2 \comp g$, then $g\comp  (h_1 \sqcap  h_2) \sqsubseteq (h_1 \sqcap h_2) \comp g$;
\item if  $g \comp h \sqsubseteq h \comp g$, then $g \comp h^\downarrow \sqsubseteq h^\downarrow \comp g$.
\end{enumerate}
\end{proposition}

\begin{remark}\label{rmk:fwcompletenesscompatibility}
\noindent
The notion of compatibility is also known in abstract interpretation as \emph{forward completeness}, see \cite{GiacobazziQ01}, which corresponds to require that no loss of precision is introduced by approximate the range of a function in a given abstract domain. This notion have been used for generalising strong preservation to abstract interpretation-based model checking~\cite{RanzatoT07}. 
\end{remark}

 \subsection{Proving soundness of equivalence closure}\label{sec:HKsoundness}
 Recall the monotone map $b\colon \Rel_X \to \Rel_X$ defined in~\eqref{eq:functional-bisim-da} and the up-closure $e\colon \Rel_X \to \Rel_X$ introduced in Section~\ref{ssec:HK}.
 In order to prove that the Hopcroft and Karp algorithm is sound one has to rely on the fact that $e$ is sound \wrt{}  $b$. 
 Thanks to Theorem~\ref{prop:compatible}, one can prove soundness by showing that $e$ is compatible \wrt{}  $b$. The proof of compatibility can be made modular using Proposition~\ref{prop:mod}.

The map $b$ can be decomposed as $b = b_* \sqcap f$ where $b_*,f \colon \Rel_X\to  \Rel_X$ are defined for all relations $R$ as
\begin{align}\label{eq:HKbf}
b_*(R) &= \{(x,y) \mid \text{ for all } a\in A, \, (t(x)(a), t(y)(a))\in R  \} \\ 
f(R) &= \{(x_1,x_2)  \mid o(x_1)=o(x_2) \}\label{eq:test}
\end{align}
and the equivalence closure as $e= (id\sqcup r \sqcup s \sqcup t )^\uparrow$ where 
$r,s,t\colon \Rel_X \to \Rel_X$ are defined as follows.
\begin{align*}
	r(R)= \{(x,x)\mid x\in X \} \qquad s(R)=\{(y,x)\mid (x,y)\in R\}\\
	t(R)=\{(x,z) \mid \exists y \text{ such that } (x,y) \in R \text{ and } (y,z)\in R\}  
\end{align*}

The proof of  compatibility of $e$ \wrt{}  $b$ can be decomposed by compatibility of $e$ \wrt{}  $b_*$ and $f$ and then use Proposition~\ref{prop:mod}.5.
Furthermore, to prove that $e$ is compatible \wrt{}  $b_*$, one can prove that $r,s,t$ are compatible \wrt{}  $b_*$ and then use points 1,3 and 4 of Proposition~\ref{prop:mod}. 
For $f$, it is immediate to check that $ef\sqsubseteq fe=f$, that is $f(R)$ is an equivalence relation for all $R\in \Rel_X$.

\begin{remark}\label{rmk:mod}
	Proving compatibility of each of $r,s,t$ is much simpler than proving compatibility of the whole $e$ at once. 
	Observe that while $e$ is an up-closure, the maps $r,s,t$ are \emph{not}. As anticipated in Remark~\ref{rmk:completenessupto}, this is the main explanation of why it is convenient to consider up-to techniques as arbitrary monotone maps, rather than just up-closures. 
\end{remark}
\begin{remark}
Some works \cite{GiacobazziRS00,GiacobazziLR15} have studied modularity for proofs of full-completeness of abstract domains, but always focusing on up-closures rather than on monotone maps. This example, together with the results in Section~\ref{sec:bridge}, shows that also for abstract domains could be convenient to decompose up-closures into smaller monotone maps. Indeed, the pointwise least-upper bound of closure operators is not necessarily a closure, but a mere monotone map. Thanks to Proposition~\ref{prop:mod}.4, one can first makes modular proofs with monotone maps and then transforms them through 
the operator \( (\cdot)^{\uparrow} \) into up-closures.
\end{remark}

\subsection{Completeness and the domain of signs}\label{sec:Signcomplete}
Recall the domain of signs in Section~\ref{ssec:AItoy} and $b\colon \Pred_Z \to \Pred_Z$ defined in~\eqref{eq:babstractInt}. One would like to prove that $s$ is complete \wrt{}  $b$ by mean of Theorem~\ref{prop:Cousot}, namely by proving that $s$ is fully complete. But this approach does not work.

For later use, it is convenient to decompose $b$ as $i \sqcup b^*$ where $i,b^* \colon \Pred_Z \to  \Pred_Z$ are defined for all predicates $P$ as follows:
\begin{equation}\label{eq:toyif}
i(P) = \{ 5 \}\qquad b^*(P) = (P\cap [1,\infty) ) \ominus\! 1 \enspace .
\end{equation}
Observe that $s$ is fully complete \wrt{}  $i$, more generally any abstract domain $a$ is fully complete with any constant function $c$ (that is $c\sqsubseteq a(c)$), but \emph{not} \wrt{}  $b^*$. To see the latter, take for instance $x=\{3\}$, and observe that $b^*s(x)=[0, \infty)$ and $sb^*(x)=[1,\infty)$, that is 
\begin{equation}\label{eq:notfc}
b^*s \not \sqsubseteq sb^*\enspace .
\end{equation}
The same $x$ shows that $(i\sqcup b^*)s \not \sqsubseteq s(i\sqcup b^*)$.

\section{Relating Abstract Interpretation and Coinduction up-to by adjointness}\label{sec:bridge}
So far, we have seen that coinduction up-to and abstract interpretation exploit a closure operator $a$ to check, respectively,  $i \sqsubseteq \nu b$ and $\mu b \sqsubseteq f$ for some $b\colon C \to C$ and $i,f \in C$. To relate them, hereafter we assume, for coinduction up-to, that $b= b_* \sqcap f$ and, for abstract interpretation, that $b=i \sqcup b^*$ where $b^*$ and $b_*$ are left and right adjoint. Intuitively, the elements of $C$ represent some predicates, or conditions, $i$ and $f$ initial and final conditions, and $b^*$ and $b_*$ predicate transformers mapping a condition into, respectively, its strongest postcondition and weakest precondition. (Note that above and hereafter we implicitly identify $i,f\in C$ with the constant maps $i,f\colon C \to C$).

In this setting, the problems addressed by coinduction up-to and abstract interpretation, namely $i \sqsubseteq \nu (b^* \sqcap f)$ and $\mu (b_*\sqcup i) \sqsubseteq f$, coincide as shown by the following well-known fact.

\begin{proposition}\label{prop:correspondencefixedpoints}Let $C \galois{b^*}{b_*} C$ and $i,f,x\in C$. 
\[(b^*\sqcup i)(x) \sqsubseteq x \sqsubseteq f \qquad \text{ if{}f } \qquad 
i \sqsubseteq x \sqsubseteq (b_*\sqcap f) (x)\]
\end{proposition}

The result below follows immediately by Knaster-Tarski.

\begin{corollary}\label{cor:coincidence}
$\mu (b^*\sqcup i)\sqsubseteq f$ if{}f $i \sqsubseteq \nu (b_*\sqcap f) $.
\end{corollary}

Our key observation is that, in this case, the sufficient conditions ensuring soundness of up-to techniques --compatibility-- and completeness of abstract interpretation --full completeness-- are closely related. Indeed, as already noticed in  \cite[Lemma~4.2]{GiacobazziRS00}, it is straightforward to see that whenever $b^*$ and $b_*$ are adjoint and $a$ is an up-closure:
\begin{equation}\label{eq:bridge}
b^*a \sqsubseteq ab^*\quad \text{ if{}f } \quad ab_* \sqsubseteq b_*a \enspace .
\end{equation}
Full completeness of $a$ \wrt{}  $i\sqcup b^*$ amounts to $(b^* \sqcup i) a \sqsubseteq a (b^* \sqcup i)$
which, by Proposition~\ref{prop:mod}.3, is entailed by 
\begin{equation}\label{fori}
b^*a \sqsubseteq a b^*  \quad \text{ and } \quad i \sqsubseteq ai \enspace .
\end{equation}
Compatibility of $a$ \wrt{}  $b_*\sqcap f$ amounts to 
\( a(b_*\sqcap f) \sqsubseteq (b_* \sqcap f) a \)
which, by Proposition~\ref{prop:mod}.5, is entailed by 
\begin{equation}\label{backf}
ab_* \sqsubseteq b_*  a \quad \text{ and } \quad af \sqsubseteq f \enspace .
\end{equation}
Observe that there is a slight asymmetry in~\eqref{fori} and~\eqref{backf}: the condition $i \sqsubseteq ai $ is guaranteed for any $i\in C$, since $a$ is an up-closure. This  is not the case for $af \sqsubseteq f $. But the latter condition is anyway necessary to make abstract interpretation meaningful (see e.g. Lemma~\ref{lemma:absouncomplete}). In this setting, whenever $a$ satisfies one of the two equivalent formulations of~\eqref{eq:bridge}, $a$ can be regarded as both a complete abstract domain for $(i \sqcup b^*)$ and a sound up-to technique for $(b_* \sqcap f)$.
This discussion is summarised below.
\begin{assumption}\label{assumption}
Consider:
\begin{inparaenum}[\upshape(\itshape i\upshape)]
\item a complete lattice $C$,
\item a pair of adjoint $C \galois{b^*}{b_*} C$,
\item a closure operator $a\colon C \to C$,
\item an element $i\in C$, and
\item an element $f\in Pre(a)$.
\end{inparaenum}
\end{assumption}

\begin{theorem}\label{thm:link}
Under Assumption~\ref{assumption}, if $a$ satisfies one of the two equivalent formulations of~\eqref{eq:bridge}, then $a$ is both a complete abstract domain for $(i \sqcup b^*)$ and  a sound up-to technique for $(b_* \sqcap f)$.
\end{theorem}

\subsection{Hopcroft and Karp from the standpoint of abstract interpretation}\label{sec:standpoint}
We now show how the Hopcroft and Karp algorithm~\cite{Aho:1974:DAC:578775,HopcroftKarp} can be seen as an instance of complete abstract interpretation, using the technology developed above.

Recall from Section~\ref{sec:HKsoundness} that $b\colon \Rel_X\to \Rel_X$ in~\eqref{eq:functional-bisim-da} can be decomposed as
$b = b_* \sqcap f$ for $b_*$ and $f$ as in~\eqref{eq:HKbf} and in \eqref{eq:test}.
The map $b_*$ has a left adjoint $b^*\colon \Rel_X \to \Rel_X$ defined for all relations $R$ as
\begin{multline}\label{eq:bstarDA}
	b^*(R) = \{(x_1',x_2') \mid \exists (x_1,x_2)\in R, a\in A \text{ such that } \\
	t(x_1)(a)=x_1' \text{ and } t(x_2)(a)=x_2' \}\enspace .
\end{multline}
To sum up we have: \( (\Rel_X,\subseteq) \galois{b^*}{b_*} (\Rel_X,\subseteq)\).\footnote{This  is similar to  the \(post\) and \(\widetilde{pre}\) operator given by \citeauthor{Cousot2000}~\cite[Example~3]{Cousot2000}}

We take $i$ as $\{(x_1,x_2)\}$, i.e., the states to prove to be language equivalent. 
By Corollary~\ref{cor:coincidence}, checking $\{(x_1,x_2)\} \sqsubseteq \mathord{\sim} = \nu (b_* \sqcap f)$ is equivalent to checking $\mu(i\sqcup b^*) \sqsubseteq f$. This inequality has a rather intuitive meaning: $\mu(i\sqcup b^*)$ is the sets of pairs of states that are ``reachable'' from the initial pair $i$. Clearly, $x_1 \sim x_2$ if{}f each of these pairs of states is in $f$.

In Section~\ref{sec:HKsoundness}, we have shown that the equivalence closure $e\colon \Rel_X \to \Rel_X$ is a sound up-to technique, by proving that $eb_* \sqsubseteq b_*e$ and $ef\sqsubseteq f$. 
By~\eqref{eq:bridge}, $b^*e \sqsubseteq eb^*$ and since $i\sqsubseteq ei$, by Proposition~\ref{prop:mod}.3, one has that $(i \sqcup b^*)e \sqsubseteq e(i \sqcup b^*)$, that is $e$ is fully complete \wrt{}  $(i \sqcup b^*)$. 
Therefore $e$ is a complete abstract domain for  $(i \sqcup b^*)$ and $f\in Pre(e)$.

\medskip

This provides a novel perspective on the Hopcroft and Karp's algorithm \cite{Aho:1974:DAC:578775}. Its correctness can be established using the least fixed-point of the function $b^*\sqcup i$ abstracted to the lattice of equivalence relations $ERel_X$. We denote this function by $\overline{b^*\sqcup i}^e\colon ERel_X \to ERel_X$ and $(Rel_X,\subseteq) \galoiS{\alpha}{\gamma} (ERel_X,\subseteq)$ the pair of adjoint associated to $e$. The map $\alpha$ assigns to every relation its equivalence closure; $\gamma$ is just the obvious injection. The function $\overline{b^*\sqcup i}^e$ is given by $\alpha \comp (b^*\sqcup i) \comp \gamma$. The algorithm returns true if{}f $\mu (\overline{b^*\sqcup i}^e) \sqsubseteq f$ as we show later on.

This leads to a slightly different algorithm, illustrated in Fig~\ref{fig:HKAI}, than the one discussed in Section~\ref{ssec:HK}: during the while loop \texttt{(3)}, it is not checked whether $o(x_1')\neq o(x_2')$ (step \texttt{(3.3)} in Fig~\ref{fig:naive}), but this is done only at the very end for the computed relation $R$ (step \texttt{(4)} in Fig~\ref{fig:HKAI}). Moreover, after every iteration of the while loop in Fig~\ref{fig:HKAI}, $R$ is an equivalence relation, while in the other algorithm $R$ is always a mere relation.
\begin{remark}
	The original algorithm by Hopcroft and Karp \cite{HopcroftKarp,Aho:1974:DAC:578775} is actually the one in Fig~\ref{fig:HKAI}: indeed, they use the so called \emph{union-find} data structure for the \emph{equivalence} relation $R$ and they check containment in $f$ only at the end.
\end{remark}

\begin{figure}[t]
\centering
\underline{$\texttt{HK}$ $(x_1,x_2)$}
\begin{codeNT}
(1) $R := \emptyset$; $todo := \emptyset$
(2) insert $(x_1,x_2)$ into $todo$
(3) while $todo$ is not empty do 
   (3.1)  extract $(x_1',x_2')$ from $todo$
   (3.2)  if $(x_1',x_2')\in R$ then goto (3)
   (3.3)  for all $a\in A$, 
	           insert $(t(x_1')(a),\,t(x_2')(a))$ into $todo$
   (3.4)  insert $(x_1',x_2')$ into $R$ 
   (3.5) $R:= e(R)$;
(4) return $R \subseteq f$; 
\end{codeNT}
\caption{Hopcroft and Karp's algorithm \cite{Aho:1974:DAC:578775}.}
\label{fig:HKAI}
\end{figure}

To be completely formal, we must say that the algorithm does not compute exactly the chain for $\mu (\overline{b^*\sqcup i}^e)$.
\[\bot \sqsubseteq \overline{b^*\sqcup i}^e(\bot) \sqsubseteq \overline{b^*\sqcup i}^e(\overline{b^*\sqcup i}^e(\bot)) \sqsubseteq \cdots \]
Indeed, at every iteration of the algorithm, only \emph{one} pair of states is removed from $todo$ and inserted into $R$, while in the above chain \emph{many} pairs are added at the same time. However, the final result, i.e., the relation $R$ at step $\texttt{(4)}$, hereafter denoted as $R^\texttt{(4)}$ is exactly $\mu (\overline{b^*\sqcup i}^e)$. From this fact and the fact the $e$ is a complete abstract domain for $b^*\sqcup i$ it follows that the algorithm is sound and complete.

In order to prove that $R^\texttt{(4)} = \mu (\overline{b^*\sqcup i}^e)$, we first show that $R^\texttt{(4)}$ is a fixed point of $\overline{b^*\sqcup i}^e$.
Observe that at step $\texttt(3)$, it always holds that
\begin{equation}\label{invHK} 
\overline{b^*\sqcup i}^e(R)=e(R\sqcup todo)\enspace .
\end{equation}
This is true after step \texttt{(2)}: $R=\emptyset$ and $todo=i$, so $\overline{b^*\sqcup i}^e(\emptyset)=e(b^*(\emptyset) \sqcup i) = e(i)=e(\emptyset \sqcup todo)$.
At any iteration, a pair $(x_1',x_2')$ is removed from \(todo\) and, if it already belongs to $R$, the control comes back to \texttt{(3)}: in this case~\eqref{invHK} is not modified.
If it does not, $(x_1',x_2')$ is inserted in $R$ and exactly $b^*(\{(x_1',x_2')\})$ is inserted in $todo$: in this case we need to check that
\begin{equation}\label{invHK2}
\overline{b^*\sqcup i}^e(R\sqcup \{(x_1',x_2')\})=e(R\sqcup todo \sqcup b^*(\{(x_1',x_2')\}))
\end{equation} 
It is easy to see that~\eqref{invHK} entails~\eqref{invHK2}:
\begin{align*}
	& \overline{b^*\sqcup i}^e(R\sqcup \{(x_1',x_2')\}) \\
	=& e(b^*\sqcup i( R\sqcup \{(x_1',x_2') \}) ) & \text{by definition of }\overline{b^*\sqcup i}^e\\
	=& e( b^*\sqcup i( R) \sqcup b^*(\{(x_1',x_2')\})) & b^* \text{ is a left adjoint}\\
	=& e( e(R \sqcup todo) \sqcup b^*(\{(x_1',x_2')\})) & \text{ by }~\eqref{invHK} \\
	=& e( R \sqcup todo \sqcup b^*(\{(x_1',x_2')\}) )  & e \text{ is a closure}
\end{align*}

Now, at step $\texttt{(4)}$, $todo$ is empty and, thus by~\eqref{invHK}, we have that $\overline{b^*\sqcup i}^e(R^{\texttt{(4)}})=e(R^{\texttt{(4)}})=R^{\texttt{(4)}}$. This proves that $R^{\texttt{(4)}}$ is a fixed point of $\overline{b^*\sqcup i}^e$.

To prove that $R^{\texttt{(4)}}= \mu \overline{b^*\sqcup i}^e$, is now enough to show that $R^{\texttt{(4)}}\sqsubseteq \mu \overline{b^*\sqcup i}^e$. Let $R_j$ and $todo_j$ be the relations at step \texttt{(3)} at the $j$-th iteration. Then, a simple inductive argument confirms that $todo^j \sqsubseteq (\overline{b^*\sqcup i}^e)^{j+1}(\bot)$ and $R_j \sqsubseteq (\overline{b^*\sqcup i}^e)^j(\bot)$ . Therefore $R^{\texttt{(4)}}=\bigsqcup R_j \sqsubseteq \bigsqcup{(\overline{b^*\sqcup i}^e)^j}(\bot) = \mu  (\overline{b^*\sqcup i}^e)$.

\subsection{The domain of signs as an up-to technique}\label{sec:signupto} 

Recall the toy program from Section~\ref{sec:toyprogramanalysis}. One needs to check whether $\mu (i \sqcup b^*) \sqsubseteq f$ where $i$ and $b^*$ are as in~\eqref{eq:toyif} and $f=[0,\infty)$. The right adjoint of $b^*$ is $b_*\colon Pred_Z\to Pred_Z$ defined for all predicates $P$ as
\begin{equation}\label{eq:signcoind}b_*(P) = \bigcup \{Q \mid b^*(Q)\subseteq P\} = ( (-\infty,0]\cup P) \oplus\!1
\end{equation}
where $\oplus \!1 \colon Pred_Z \to Pred_Z$ is defined as $P\oplus \!1 = \{i+1\mid i\in P  \}$.

Thanks to Corollary~\ref{cor:coincidence}, rather than checking $\mu (i \sqcup b^*) \sqsubseteq f$, one can check $i \sqsubseteq \nu (b_* \sqcap f)$. The latter can be proved by means of coinduction: one has  to find a predicate $P$ such that $\{5\} \subseteq P \subseteq b_*(P) \cap [0,\infty)$. For instance, by taking $P = \{5,4,3\}$, one has $b_*(P) = (-\infty,1] \cup \{6,5,4\}$ and $b_*(P) \cap [0,\infty) = \{6,5,4,1,0\}$. 
Therefore the inclusion does not hold. In order to find a $(b_* \sqcap f)$-simulation $P$, one can take the least fixed point computed in~\eqref{eq:lfpcomp}, that is  $P= \{5,4,3,2,1,0\}$.

\medskip

One can also reason, more effectively, up-to the abstract domain of signs $s$ (Section~\ref{ssec:AItoy}). In this case, $\{5\}$ itself is a $(b_* \sqcap f)$-simulation up to $s$. Indeed, $s(\{5\}) = [1,\infty)$ and  $ b_*[1,\infty) \cap f = ( ( (-\infty, 0] \cup [1,\infty) ) \oplus \! 1)  \cap [0,\infty) = [0,\infty)$. Obviously $\{5\}\subseteq [0,\infty)$. 

To make this a valid proof, one should show first that $s$ is a sound up-to technique. Unfortunately~\eqref{eq:notfc} and~\eqref{eq:bridge} inform us that $sb_* \not \sqsubseteq b_* s$, namely $s$ is not $b_*$-compatible. Note that this does not entail that  $s$ is not $(b_*\sqcap f)$-compatible, but by taking $x=\{-3\}$ one can easily verify that this is the case, i.e., 
\[
s(b_*\sqcap f) \not \sqsubseteq (b_*\sqcap f)s\enspace .
\] 
Nevertheless, $s$ is sound \wrt{}  $(b_*\sqcap f)$: we will show this in Section~\ref{sec:companion}.

\section{Intermezzo}\label{sec:intermezzo}
Before continuing with the next achievements, we make two small detours to settle down the concepts seen so far.

\subsection{A counterexample to the correspondence of soundness and completeness}\label{sec:counter}
In Section~\ref{sec:bridge}, we have shown that the conditions in~\eqref{eq:bridge} entails both soundness of up-to techniques and completeness of abstract domains.
The reader may wonder whether, more generally, it is the case that an up-to technique is sound if{}f it is a complete abstract domain.

More formally, given Assumption~\ref{assumption} is it the case that $a$ is $(b_*\sqcap f)$-sound (as an up-to technique) if{}f it is $(i \sqcup b^*)$-complete (as an abstract domain)?

The answer is no.  Consider the following lattice, with $b^*$ defined by the dashed lines on the left and $b_*$ defined by the dotted lines on the right. It is easy to check that they are adjoint. Take $i=1$ and $f=4$. 
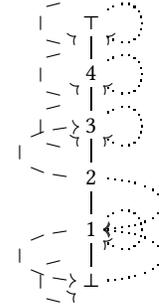
\[
\begin{array}{c}
{\lower-1.2cm\hbox{\xymatrix@R=0.3cm@C=0.3cm{  
\top \ar@{-->}@(ul,dl) \ar@{.>}@(ur,dr)  \\
4 \ar@{-->}@(ul,dl) \ar@{-}[u] \ar@{.>}@(ur,dr)   \\
3 \ar@{-->}@(ul,dl) \ar@{-}[u]  \ar@{.>}@(ur,dr)  \\
2\ar@{-}[u] \ar@{-->}@(l,l)[u] \ar@{.>}@(r,r)[d] \\
1\ar@{-}[u] \ar@{-->}@(l,l)[d] \ar@{.>}@(ur,dr)  \\
\bot \ar@{-->}@(ul,dl) \ar@{-}[u] \ar@{.>}@(r,r)[u]
}}}\end{array}
\]

Let $a$ be the up-closure  such that $Pre(a)= \{\top, 4,3,2\}$.
Then $\mu (i \sqcup b^*) = 1$ and $a(\mu (i \sqcup b^*) ) =2$.
Instead $\mu a (i \sqcup b^*)= 3$. Indeed,
\[\bot \sqsubseteq a (i \sqcup b^*)(\bot)= 2 \sqsubseteq a (i \sqcup b^*)(2)=3 \sqsubseteq a (i \sqcup b^*)(3)=3\enspace .\]
Therefore, by Lemma~\ref{lemma:alternativecompleteness}, $a$ is not complete \wrt{}  $i \sqcup b^*$.
However, $a$ is $(b_*\sqcap f)$-sound. Indeed $\nu (b_*\sqcap f)$ is computed as 
\begin{equation}\label{eq:counterexOmega}
\top \sqsupseteq 4 \sqsupseteq 4
\end{equation}
and, similarly, $\nu (b_*\sqcap f)a = 4$.

\subsection{Duality}\label{sec:duality}
The reader may have got the feeling that coinduction up-to and abstract interpretation are somehow the dual of each other. This is not the case: first, coinduction up-to is a proof technique, exploiting the Knaster-Tarski fixed point theorem, while abstract interpretation is a computational method relying on Kleene's theorem; second both abstract interpretation and coinduction up-to use as enhancement an up-closure $a\colon C \to C$, while their duals should use down-closures. The latter is explained in
some details, below.

 The dual of the coinduction up-to looks like
\begin{equation}\label{eq:inductionuptoproofprinciple}
 \begin{array}{c}
    \exists y, \;  ba(y) \sqsubseteq y\sqsubseteq x \\
    \hline 
    \mu b \sqsubseteq x
\end{array}
\end{equation}
When $a$ is an up-closure, this principle does not provide any enhancement \wrt{}  standard induction (see~\eqref{eq:coinductionproofprinciple}, left): indeed, if 
$ba(y) \sqsubseteq y$, then also $b(y) \sqsubseteq y$. Instead, when $a$ is an down-closure, the principle might be meaningful.

\medskip

The dual of abstract interpretation consists in checking $\nu b\sqsubseteq f$ by optimising somehow the computation of the chain of $\nu b$ in the right of~\eqref{eq:initfinsequences}.
Interestingly enough, all the elements of this chains already belongs to the domain $Pre(a)$, whenever $a$ is a fully-complete up-closure.

\begin{proposition}\label{prop:dual}
	Let $C \galois{b^*}{b_*} C$ and $i,f\in C$. Let $a\colon C\to C$ be an up-closure fully complete \wrt{} $b^*$. Assume moreover that $a(f)\sqsubseteq f$. For all $k$, \[a(b_*\sqcap f)^k(\top) \sqsubseteq (b_*\sqcap f)^k(\top)\enspace .\]
\end{proposition}

This provides an explanation for what we anticipated in Remark~\ref{rmk:hopcroft}. Indeed, we have seen in Section~\ref{sec:HKsoundness} that the equivalence closure $e\colon \Rel_X \to \Rel_X$ is fully complete \wrt{}  $b_*$ in~\eqref{eq:HKbf}; the above proposition states that all the elements of the chain~\eqref{eq:initfinsequences} for computing $\nu (b_*\sqcap f)$ are in $Pre(e)$, i.e., they are equivalence relations.  

It is worth to conclude this detour on duality, by remarking that while abstract interpretation and coinduction up-to naturally emerges in logics, computer science and related fields, their duals, exploiting a down-closure operator, are far less common.

\section{The companion}\label{sec:companion}
In Section~\ref{sec:signupto}, we have seen that the domain of signs $s$  is not compatible \wrt{}  $(b_*\sqcap f)$. Nevertheless, we will see at the end of this section that $s$ is sound.
The strategy that we are going to use to prove this fact exploits recent developments in up-to techniques~\cite{hur2013power,pous2016coinduction,parrow2016largest} that, in the next section we will transfer to abstract interpretation. The proof strategy is based on the following observations:
\begin{enumerate}
\item The class of sound up-to techniques is downward closed: if $a_1\sqsubseteq a_2$ and $a_2$ is $b$-sound, then also $a_1$ is $b$-sound.
\item Fixed a $b$, there exists a greatest $b$-compatible up-to technique $\omega_b$, which \citeauthor{pous2016coinduction}~\cite{pous2016coinduction} calls \emph{the companion}.
\end{enumerate}
Therefore, rather than proving that a certain up-to technique $a$ is compatible, to show the soundness of $a$ is enough to prove that $a \sqsubseteq \omega_b$. 
This is extremely useful because there are many techniques which are not compatible, but still they are below the companion (and thus sound), like for instance the domain of signs from Section~\ref{sec:signupto} or many of the so called respectful techniques \cite{San98MFCS} which are common in process calculi and GSOS specifications.

Interestingly enough, $\omega_b$ is an up-closure also when one considers as up-to techniques arbitrary monotone maps, rather than just up-closure operators (see Lemma 3.2 \cite{pous2016coinduction}).
This fact allows us to give an alternative characterisation of $\omega_b$ as an abstract domain $\Omega_b$ which we found suggestive and useful (at least in our examples), but that can be easily derived from the results of \citeauthor{pous2016coinduction}~\cite{pous2016coinduction}. We first need the following well-known lemma from~\cite{Ward42}.

\begin{lemma}\label{lemma:inclusion}
Let $a_1,a_2 \colon C \to C$ be two up-closures.
$a_1\sqsubseteq a_2$ if{}f $Pre(a_2) \subseteq Pre(a_1)$.
\end{lemma}

\begin{theorem}\label{thmcompanion}
Let $b\colon C \to C$ be a Scott cocontinuous map\footnote{The assumption of Scott co-continuity can be removed to the price of a more elaborated characterization of $\Omega_b$. } and let $\omega_b\colon C \to C$ be the closure operator associated to the sublattice $\Omega_b$, defined as follows. \[\top \sqsupseteq b(\top)\sqsupseteq bb(\top) \sqsupseteq \dots \sqsupseteq \nu b\]
Then $\omega_b$ is the greatest $b$-compatible map, that is 
\begin{inparaenum}[\upshape(\itshape 1\upshape)]
  \item $\omega$ is compatible and
	\item if $a$ is compatible, then $a\sqsubseteq \omega_b$.
\end{inparaenum}
Moreover, (3) for any up-closure $a$, $a\sqsubseteq \omega_b$ if{}f $\Omega_b \subseteq Pre(a)$.
 \end{theorem}
The theorem helps in understanding the difference between being compatible and being below the companion. By point $\bullet$3 of Lemma~\ref{lemma:equivalentformulation}, $a$ is compatible if{}f $(A=)Pre(a)$ is closed by $b$, that is for all $x\in Pre(a)$, $b(x)\in Pre(a)$. Being below the companion means instead that just $\Omega_b$ should be included into $Pre(a)$. This latter condition is obviously much weaker, but still is enough to entail soundness.
\begin{corollary}\label{cor:silly}
If $b^i(\top)\in Pre(a)$ for all $i \in \Nat$, then $a$ is $b$-sound.
\end{corollary}
This corollary  is not particularly useful to prove soundness, since the premise is often hard to check. However, this is enough for the purposes of our paper. The condition of being below the companion could be better checked by defining the companion itself as the greatest fixed point of a certain ``second order'' operator and then use again coinduction. We stop here, as this goes beyond the scope of this paper, and we refer the interested reader to the work of \citeauthor{pous2016coinduction}~\cite{pous2016coinduction}. It is however important to remark here that, in Section~\ref{sec:loccoind}, we will give a coinductive characterization for an analogous of the companion in the context of abstract interpretation.

\begin{example}
We conclude this section, by showing that the domain of signs $s$ is a sound up-to technique for $(b_*\sqcap f)$. In this case, it is easy to compute $\Omega_b$: 
\begin{equation*}
Z  \sqsupseteq [0,\infty) 
\end{equation*}
Since this is included into $Pre(s)$, which is the domain in Section~\ref{ssec:AItoy}, then by Corollary~\ref{cor:silly}, $s$ is sound. Note instead that the domain of signs $Pre(s)$ is \emph{not} closed under $b_*\sqcap f$: this means exactly that $s$ is not $b_*\sqcap f$-compatible.
\end{example}

\begin{remark}
	The existence of the  smallest abstract domain (or equivalently by virtue of Lemma \ref{lemma:inclusion} the greatest up-closure) that is fully complete \wrt{} $i\sqcup b^*$ is irrelevant for abstract interpretation because it is always  the abstract domain containing only $\top$ that, obviously, does not contain the property $f$ which needs to be checked. However, it makes sense to look for, amongst all the abstract domains containing $f$, the smallest fully complete one. The $f$-companion that we will introduce in the next section is the smallest abstract domain (or equivalently the largest up-closure) containing $f$ that is fully complete \wrt{} $b^*$ (by Proposition \ref{prop:mod}.3 this is also fully complete \wrt{} $i \sqcup b^* $).
\end{remark}

\section{Local Completeness}\label{sec:localcompleteness}
Inspired by up-to techniques, we give a novel definition of completeness, called local completeness. This notion is strictly weaker than completeness, but still is sufficient to solve the
original problem of program analysis, namely to check whether $\mu b \sqsubseteq f$ for a given property $f$ and predicate transformer $b$. 

\begin{definition}
Let $C$ be a complete lattice, $b\colon C \to C$ be a monotone map and $f\in C$. We say that an up-closure $a\colon C \to C$ is \emph{local complete}, or \emph{$(b,f)$-complete}, if{}f
\begin{inparaenum}[\upshape(\itshape 1\upshape)]
  \item  $a(f)\sqsubseteq f $  and 
	\item  $ \mu (a b) \sqsubseteq f \text{ if{}f } \mu b \sqsubseteq f$.
\end{inparaenum}
\end{definition}

Our interest in $(b,f)$-completeness is justified by the following result, stating that, rather than checking $\mu b \sqsubseteq_C f$, one can safely lift $b$ to the abstract domain $A=Pre(a)$ and check whether $ \mu  \overline{b}^a \sqsubseteq_A f$.
\begin{proposition}\label{prop:localimplies}
If $a$ is $(b,f)$-complete, then  $ \mu  \overline{b}^a \sqsubseteq_A f$ if{}f  $\mu b \sqsubseteq_C f$.
\end{proposition}

We named $(b,f)$-completeness also local completeness since, as illustrated by the following result, it is similar to completeness but localised at $f$.
\begin{proposition}\label{prop:completeifflocal}
Let $C$ be a complete lattice, $b,a\colon C \to C$ be a monotone map and a an up-closure.
Then:
\begin{center}
$a$ is $b$-complete \\ if{}f \\ for all $f\in Pre(a)$, $a$ is $(b,f)$-complete. 
\end{center}
\end{proposition}
\begin{example}\label{ex:localvsglobal}
Consider $a$, $b^*$, $i$ and the lattice in Section~\ref{sec:counter} and recall that $a$ is not $(i \sqcup b^*)$-complete. 
Now for $f\in  \{\top, 4,3\}$, it is immediate to see that $a$ is $(i \sqcup b^*,f)$-complete. Instead, for $f=2$, it is not: $\nu(i \sqcup b^*) =1$, while $\nu a (i \sqcup b^*)=3$.
\end{example}

In Section~\ref{sec:companion}, we have seen that the class of sound up-to techniques is downward closed. This is not the case with the standard definition of completeness for abstract domains (Example~\ref{ex:notdown}) but it holds for local completeness (Proposition~\ref{prop:bfclosed}).

\begin{example}\label{ex:notdown}
Recall from Example~\ref{ex:localvsglobal} that $a$  is not $(i \sqcup b^*)$-complete. 
Now take $a'$ to be the up-closure  such that $Pre(a')= \{\top, 4,3\}$. In this case,  one has that $a'(\mu (i \sqcup b^*) ) = 3 = \mu a' (i \sqcup b^*)$, i.e., $a'$ is $(i \sqcup b^*)$-complete. In this case $a\sqsubseteq a'$.
\end{example}

\begin{proposition}\label{prop:bfclosed} Let $C$ be a complete lattice and  $b,a_1,a_2\colon C \to C$ be a monotone map and two up-closures such that $a_1\sqsubseteq a_2$. Let $f\in C$. 

\noindent If $a_2$ is $(b,f)$-complete, then $a_1$ is $(b,f)$-complete.
\end{proposition}

This property makes the proof of local completeness much easier than those of completeness. Indeed, for the latter it is enough to prove full completeness, while for the former it is enough to prove to be below some fully complete domain (Theorem~\ref{them:localcomp}). This is similar to what happens with the companion for up-to techniques. However, the small asymmetry of $i$ and $f$ discussed in Section~\ref{sec:bridge} forces us to consider a little variation of the notion of companion.

\begin{definition}
Let $C$ be a complete lattice, $b\colon C \to C$ be a monotone map and $f\in C$. 
A monotone map $a\colon C \to C$ is \emph{\((b,f)\)-compatible}, if{}f
\begin{inparaenum}[\upshape(\itshape 1\upshape)]
  \item  \(a(f)\sqsubseteq f \)  and 
	\item  \(a b \sqsubseteq ba\).
\end{inparaenum}
The \emph{$f$-companion} of $b$ is 
\[\omega_{b,f}=\bigsqcup \{a \mid a \text{ is $(b,f)$-compatible} \}\enspace .\] 
\end{definition}
In the above definition the least upper bound is taken in the lattice of monotone functions. However $\omega_{b,f}$ is guaranteed to be an up-closure which, additionally, is $(b,f)$-compatible. 
\begin{proposition}\label{prop:fcompanion}
The following holds:
\begin{enumerate}
\item $\omega_{b,f} b \sqsubseteq b \omega_{b,f}$, 
\item $\omega_{b,f}(f)\sqsubseteq f$,
\item $x \sqsubseteq \omega_{b,f}(x)$ for all $x\in C$,
\item $\omega_{b,f}(\omega_{b,f}(x)) \sqsubseteq \omega_{b,f}(x)$ for all $x\in C$.
\end{enumerate}
\end{proposition}

Observe that $(b,f)$-compatibility entails $(b\sqcap f)$-compatibility by Proposition~\ref{prop:mod}.5, but the converse does not hold in general. We need this stronger notion of compatibility because, under Assumption~\ref{assumption}, $(b_*\sqcap f)$-compatibility alone does not allow to deduce $(i\sqcup b^*)$-completeness. Instead,  for $(b_*,f)$-compatibility, this follows immediately from Theorem~\ref{thm:link}.
\begin{corollary}\label{cor:bfcomptivility}
	Let $C \galois{b^*}{b_*} C$ be a pair of adjoint, $a\colon C \to C$ be an up-closure and $i,f\in C$.  If $a$ is $(b_*,f)$-compatible, then $a$ is \((i\sqcup b^*)\)-complete, hence also \((i\sqcup b^*, f)\)-complete.
\end{corollary}
The second part of the statement follows from Proposition~\ref{prop:completeifflocal}.
Next, we combine Proposition~\ref{prop:bfclosed}, Proposition~\ref{prop:fcompanion} and Corollary~\ref{cor:bfcomptivility} to obtain the main result of this section.

\begin{theorem}\label{them:localcomp}
	Let $C \galois{b^*}{b_*} C$ be a pair of adjoint, $a\colon C \to C$ be an up-closure and $i,f\in C$. 
If $a\sqsubseteq \omega_{b_*,f}$, then $a$ is $(i\sqcup b^*, f)$-complete.
\end{theorem}

It is worth to visualise the difference between $\omega_{b_*\sqcap f}$ and $\omega_{b_*,f}$ in terms of the associated abstract domains. Under the assumptions of Theorem~\ref{thmcompanion}, $\Omega_{b_*\sqcap f}$ is the sublattice of \(C\) (consisting of a chain) given by 
\[\top \sqsupseteq (b_*\sqcap f)(\top)\sqsupseteq (b_*\sqcap f)(b_*\sqcap f)(\top) \sqsupseteq \dots \sqsupseteq \nu (b_*\sqcap f)\]
that, since $b_*(\top)=\top$ and $b_*(x\sqcap y)=b_*x \sqcap b_*y$, coincides with
\[\top \sqsupseteq f \sqsupseteq b_*(f)\sqcap f \sqsupseteq b_*b_*(f) \sqcap b_*(f)\sqcap f  \sqsupseteq \dots \sqsupseteq \nu (b_*\sqcap f)\text{.}\]
Instead $\Omega_{b_*,f}$ is the smallest meet-complete sublattice of \(C\) containing 
\begin{equation}\label{eq:Omegabf}
\top \qquad f \qquad b_*(f) \qquad b_*b_*(f) \qquad b_* b_*b_*(f)  \qquad \dots
\end{equation}

\begin{corollary}\label{corSBRA}
If $b_*^j(f)\in Pre(a)$ for all $j\in \Nat$, then $a$ is $(i\sqcup b^*, f)$-complete.
\end{corollary}

\begin{example}
Recall the abstract domain of signs $\Sign_Z$, $b_*$ defined in~\eqref{eq:signcoind} and $f=[0,\infty)$. We know that $s$ is not $(b_*,f)$-compatible because of~\eqref{eq:notfc}. However, since $b_*(f)=Z$, $\Omega_{b_*,f}$ is just the complete lattice $Z \sqsupseteq [0,\infty)$. Therefore $s$ is below the $f$-companion.
\end{example}

\subsection{A coinductive characterization of the \texorpdfstring{$f$}{f}-companion}\label{sec:loccoind}
As mentioned in Section~\ref{sec:companion}, the companion enjoys a coinductive characterization that is useful  to prove by ``second order coinduction'' soundness of up-to techniques. We conclude this section by briefly showing that a similar characterization can be given for the $f$-companion in order to prove local completeness of abstract domains. Our argument is a tiny variation of Section~6 in~\cite{pous2016coinduction}.

\begin{definition}
	Let $[C \to C]$ be the complete lattice of monotone maps on $C$. The function $B\colon [C \to C] \to [C \to C]$ is defined for all $a\colon C \to C$ as \[B(a)=\bigsqcup \{c \mid c b\sqsubseteq ba, \\ c(f)\sqsubseteq f \}\enspace .\]
\end{definition}
\begin{lemma}\label{lemma:coinfcomp}
$B$ is monotone and for all functions $a,a'\colon C \to C$,
$$a'\sqsubseteq B(a) \text{ if{}f } a'b\sqsubseteq ba \text{ and } a'(f)\sqsubseteq f \text{.}$$
\end{lemma}

This means that $a$ is $(b,f)$-compatible if{}f $a\sqsubseteq B(a)$, that is $a$ is a post-fixed point of $B$. By the Knaster-Tarski fixed point theorem, one has immediately the following result.
\begin{theorem}
	$\omega_{b,f}=\nu B$
\end{theorem}

\section{Conclusion}
In this paper we studied the relationship existing in between sound up-to techniques and complete abstract domains. In general, the two concepts do not coincide (Section~\ref{sec:counter}) but, under reasonable assumptions (Assumption \ref{assumption}), the sufficient conditions that are commonly used to prove soundness of up-to techniques --compatibility-- and completeness of abstract domains --full complete\-ness-- are equivalent (Theorem~\ref{thm:link}). This allows to look at fully complete abstract domains as sound up-to techniques and, vice versa, to look at compatible up-to techniques as complete abstract domains. As an example of the latter, we have shown  that the Hopcroft and Karp's algorithm~\cite{Aho:1974:DAC:578775,HopcroftKarp}, which was recently observed to rely on up-to techniques~\cite{bp:popl13:hkc}, can also be studied from the viewpoint of complete abstract interpretation.

We hope that our observation can lead to a fruitful cross-fertilisa\-tion amongst two areas that, so far, have developed their own technologies independently. 
As a proof of concept for this technology transfer, we have shown that recent developments in up-to techniques~\cite{pous2016coinduction} lead to a weaker notion of completeness, called local completeness, that is enough to ensure that if a certain property is not satisfied in the abstract domain, then it does not hold  in the concrete one. Interestingly enough, local completeness can be proved by means of coinduction. As a short term application, we mention that, to prove completeness of an abstract domains for a certified abstract interpreter (see e.g. \cite{bertot2009structural,leroy2014formal,verasco}) one could, thanks to our work, reuse one of the many available libraries for up-to techniques that have been developed in different proof assistants (see e.g., \cite{damienlibrary,danielsson2017up}).

We leave as a future work the connection with domain completion techniques~\cite{GiacobazziRS00,CGR07,RanzatoRT08} which, intuitively, define strategies to enrich an abstract domain with new values as long as it is not precise enough to prove a given property. The correspondence between completeness in abstract interpretation and soundness in up-to techniques can also motivate the extension of methods for proving the absence of false alarms in abstract interpretations, such as the proof system in \cite{GiacobazziLR15}, to prove soundness of corresponding up-to techniques. 

\newcommand{\noopsort}[1]{}

\clearpage
\appendix

\section{A categorical perspective}\label{sec:cat}
Most of the concepts discussed in this paper can be extended from lattices to categories: a lattice can be seen as a category, a monotone map as a functor, an up-closure operator as a monad and a down-closure operator as a comonad. Pre and post-fixed point as algebras and coalgebras, the least and the greatest fixed point as the initial algebra and the final coalgebra. 

\medskip

This perspective motivates the terminology EM (Eilenberg Moore) law and Kleisli law for the conditions $\bullet 2$ and $\circ 2$ in Lemma~\ref{lemma:equivalentformulation}. Indeed, one can think to the problem of completeness of abstract interpretation and soundness of up-to techniques as the problem of \emph{extending} and \emph{lifting} the functor $b\colon C \to C$ to some functor $\overline{b}$ either on the Kleisli category $Kl(a)$ or to the Eilenberg-Moore category $EM(a)$ of algebras for the monad $a\colon C \to C$. In this case, since $C$ is a lattice, one has that $Kl(a) = EM(a) =Pre(a)$. In this perspective, completeness of full abstraction means that there is a functor $\alpha \colon Alg(b) \to Alg(\overline{b})$ preserving initial algebra (this is entailed by requiring $\alpha$ to be a left adjoint). Similarly, soundness of up-to techniques means that there is a functor $\gamma \colon Coalg(\overline{b}) \to Coalg (b)$ that preserves the final coalgebra (this is entailed by requiring $\gamma$ to be a right adjoint). The latter is rather well-studied problem, which arise for instance with bialgebras (see e.g., \cite{DBLP:journals/tcs/Klin11,turi1997towards}). The former instead is far less understood.

\begin{figure*}
\begin{tabular}{c|c}
Abstract Interpretation &  Coinduction up-to\\
\hline
Kleisli law $ba \sqsubseteq ab$ & EM-law $ab \sqsubseteq ba$  \\
Kleisli Extension $\overline{b} \colon Kl(a) \to Kl(a)$ & EM lifting $\overline{b} \colon EM(a) \to EM(a)$\\
$\alpha \colon Alg(b) \to Alg(\overline{b})$ is a left adjoint & $\gamma \colon Coalg(\overline{b}) \to Coalg (b)$ is a right adjoint \\
\end{tabular}\caption{The category theory behind complete abstract domains and sound up-to techniques}
\end{figure*}

\section{Proofs}\label{app:proof}

\subsection{Proofs of Section~\ref{sec:upto}}
\begin{proof}[Proof of Lemma~\ref{lemma:charsound}]
Suppose that~\eqref{eq:coinductionuptoproofprinciple} holds and take in its premises $i=x=\nu ba=ba(\nu ba)$. The conclusions of~\eqref{eq:coinductionuptoproofprinciple} means $\nu ba \sqsubseteq \nu b$. Vice versa, suppose that $(\dag)$ $\nu ba \sqsubseteq \nu b$. Assume the premises in~\eqref{eq:coinductionuptoproofprinciple}. By coinduction, one has that $i\sqsubseteq \nu ba$ and thus by $(\dag)$, one has that $i\sqsubseteq \nu b$. 
\end{proof}

\subsection{Proofs of Section~\ref{sec:ai}}
\begin{proof}[Proof of Lemma~\ref{lemma:absouncomplete}]
\begin{enumerate}
\item If $\mu b \sqsubseteq_C f$ then $\alpha(\mu b) \sqsubseteq_A \alpha(f) = f$ and then $a (\mu b) = \gamma\comp \alpha(\mu b) \sqsubseteq_C \gamma(f) = f$. The latter entails that $\mu b \sqsubseteq_C a(\mu b) \sqsubseteq_C f$.
\item By 1.\@ with $x=\mu b$.
\item First observe that $b\gamma(\mu \overline{b}) \sqsubseteq \gamma \overline{b}\alpha \gamma (\mu \overline{b}) = \gamma \overline{b}(\mu \overline{b})= \gamma(\mu \overline{b})$. That is $\gamma(\mu \overline{b})$ is a pre-fixed  point of $b$: thus $\mu b \sqsubseteq \gamma(\mu \overline{b})$. By adjointness $\alpha(\mu b) \sqsubseteq_A \mu \overline{b}$.
\item If $\mu \overline{b} \sqsubseteq_A f$ then, by 3., $\alpha(\mu b)  \sqsubseteq_A f$. By 2. $\mu b \sqsubseteq_C f$.
\end{enumerate}
\end{proof}

\begin{proof}[Proof of Proposition~\ref{prop:alwayssound}]
(1) trivial: $\alpha b \sqsubseteq \alpha b \gamma \alpha$. (2) Suppose that $\alpha b \sqsubseteq \overline{b} \alpha$. In particular, for all $x\in Pre(a)$, $\alpha b \gamma (x)\sqsubseteq \overline{b} \alpha \gamma(x) = \overline{b}(x)$.
\end{proof}

Therefore, for all abstract domain $a$, there exists a sound approximation $\overline{b}^a$, that is $\mu \overline{b}^a \sqsubseteq_A f$ implies that $\mu b \sqsubseteq_C f$.
The converse implication  guaranteed in general. One has to require completeness of the abstract domain: $a$ is \emph{complete} \wrt{}  $b$ (or $b$-complete, for short) if{}f $\alpha(\mu b) = \mu  \overline{b}^a$.

\begin{proof}[Proof of Lemma~\ref{lemma:characterizationcompleteness}]
Assume $\alpha(\mu b) = \mu \overline{b}^a$. By Lemma~\ref{lemma:absouncomplete}.2, $\mu \overline{b}^a \sqsubseteq_A f$ if{}f $\alpha(\mu b)  \sqsubseteq_A f$ if{}f $\mu b \sqsubseteq_C f$.
\end{proof}

\begin{proof}[Proof of Lemma~\ref{lemma:alternativecompleteness}]
This is Lemma 3.1~\cite{GiacobazziRS00}. For reader convenience, we report its proof below.

We  first prove the left-to-right implication.
The inclusion $a(\mu b) \sqsubseteq \mu (ab)$ always holds. Indeed: $b \sqsubseteq ab$, thus $\mu b \sqsubseteq \mu (ab)$ and, by monotonicity of $a$, 
$a(\mu b) \sqsubseteq a(\mu (ab))$. Moreover $a(\mu (ab))=aab(\mu(ab))=ab(\mu(ab))= \mu(ab)$.

The other inclusion holds if $a$ is $b$-complete. Indeed, in this case $\alpha(\mu b) = \mu(\alpha b \gamma)$, which entails that $\alpha b \gamma \alpha (\mu b) = \alpha(\mu b)$. By monotonicity of $\gamma$, $\gamma \alpha b \gamma \alpha (\mu b) = \gamma \alpha(\mu b)$, i.e.,  $\gamma \alpha(\mu b) = a(\mu b)$ is a fixed point of $ab$. Therefore $\mu(ab) \sqsubseteq a(\mu b)$.

\medskip

We can now prove the right-to-left implication. It always hold that $\alpha(\mu b) \sqsubseteq \mu (\alpha b \gamma)$. For the other inclusion, observe that by assumption we have $a(\mu b) = \mu (a b) $ which entails that $\gamma \alpha b \gamma \alpha (\mu b)= \gamma \alpha (\mu b)$. Since $\gamma$ is injective, $ \alpha b \gamma \alpha (\mu b)=  \alpha (\mu b)$, that is $\alpha (\mu b)$ is a fixed point of $\alpha b \gamma $. Thus $\mu(\alpha b \gamma) \sqsubseteq \alpha (\mu b)$.
\end{proof}

\subsection{Proofs of Section~\ref{sec:soundandcomplete}}
\begin{proof}[Proof of Lemma~\ref{lemma:equivalentformulation}]
First part.

($1\Rightarrow 2$) $ab = aba \sqsupseteq ba$ since $a$ is an up-closure. ($2\Rightarrow 3$) It always holds $(\alpha b \gamma) \alpha \sqsupseteq \alpha b$. For the other inclusion, observe that if $ba \sqsubseteq ab$, then $b\gamma \alpha \sqsubseteq \gamma \alpha b$ and  $\alpha b\gamma \alpha \sqsubseteq \alpha \gamma \alpha b \sqsubseteq \alpha b$.
($3 \Rightarrow 1$) Since $(\alpha b \gamma) \alpha = \alpha b$, then $\gamma (\alpha b \gamma) \alpha = \gamma\alpha b$, that is $aba = ab$.

Second part.

($1\Rightarrow 2$) Since $b\comp a = a\comp b\comp a$ then $b\comp a = a\comp b\comp a \sqsupseteq a\comp b$. 
($2\Rightarrow 1$) Since $a\comp b \sqsubseteq b\comp a$, then $a\comp b \comp a \sqsubseteq b\comp a \comp a \sqsubseteq b \comp a$. The other inclusion, $b\comp a \subseteq  a\comp b\comp a$, holds since $a$ is an up-closure.
 ($2\Rightarrow 3$) Observe that $A=Pre(a)$ by construction. For every pre fixed-point $a(x)\sqsubseteq x$, it holds that $ab(x) \sqsubseteq ba(x)\sqsubseteq b(x)$, namely $b(x)$ is a pre fixed-point. One can therefore define $\overline{b}(x) = b(x)$. The fact that $\gamma \overline{b} =  b \gamma$ follows immediately by construction of $\gamma$. 
 ($3 \Rightarrow 2$) By construction of $\gamma$, for every pre fixed point $x$, $b(x)$ is forced to be a pre fixed-point of $a$: $ab(x)\sqsubseteq b(x)$. Therefore  $ab(x)\sqsubseteq b(x) \sqsubseteq ba(x)$.
\end{proof}

\begin{proof}[Proof of Theorem~\ref{prop:compatible}]
For each $y\sqsubseteq ba(y)$, $a(y)\sqsubseteq aba(y) \sqsubseteq baa(y) \sqsubseteq ba(y)$. Therefore $a(y)\sqsubseteq \nu b$. If $x\sqsubseteq y$, then $x\sqsubseteq a(y)\sqsubseteq \nu b$.
\end{proof}

\begin{proof}[Proof of Theorem~\ref{prop:Cousot}]
The assumption of Scott-continuity is necessary to characterise 
\[\alpha(\mu b) = \alpha (\bigsqcup_{n} b^n(\bot_C) ) \qquad \text{ and } \qquad \mu(\alpha b \gamma) =  \bigsqcup_{n} (\alpha b \gamma)^n(\bot_A) \enspace .\] 
Since $\alpha$ is a left adjoint we have that the leftmost is equivalent to  $ (\bigsqcup_{n}  \alpha b^n(\bot_C) $.

By induction on $n$, we prove that $\alpha b^n(\bot_C) = (\alpha b \gamma)^n (\bot_A)$.
\begin{itemize}
\item For $n=0$, $\alpha(\bot_C) = \bot_A$;
\item For $n+1$, we have that $\alpha b \gamma (\alpha b \gamma)^n (\bot_A) = \alpha b\gamma \alpha b^n (\bot_C)$ by induction hypothesis. Using the property of Kl-lifting, the latter is equivalent to $\alpha b b^n (\bot_C) = \alpha b^{n+1}(\bot_C)$.
\end{itemize}
\end{proof}

\begin{proof}[Proof of Proposition~\ref{prop:mod}]
For the first four point see~\cite{pous:aplas07:clut} or Proposition 6.3.11~\cite{PS12}. Points 5 and point 6 follow by duality from points 3 and 4.
\end{proof}

\subsection{Proofs of Section~\ref{sec:bridge}}

\begin{proof}[Proof of Proposition~\ref{prop:correspondencefixedpoints}]
If $(b^*\sqcup i)(x) \sqsubseteq x$, then $ i \sqsubseteq x$ and $b^*x\sqsubseteq x$. From the latter, it follows that $x\sqsubseteq b_* x$. 
Since $x\sqsubseteq f$, then  $x\sqsubseteq b_* x \sqcap f$, that is $x\sqsubseteq (b_* \sqcap f)(x)$.

Conversely, $x \sqsubseteq (b_*\sqcap f)(x)$ entails that $x\sqsubseteq f$ and $x \sqsubseteq b_*x$. From the latter, it follows that $b^*x\sqsubseteq x$. 
Since $i\sqsubseteq x$, then $i\sqcup b^*x \sqsubseteq x$, that is $(i \sqcup b^*)(x) \sqsubseteq x$.
\end{proof}

\subsection{Proofs of Section~\ref{sec:intermezzo}}

\begin{proof}[Proof of Proposition~\ref{prop:dual}]
By induction on $k$. For $k=0$, the above inequality amounts to  $\top \sqsubseteq \top$, which trivially holds.
For $k+1$, $a(b_*\sqcap f)^{k+1}(\top)=a(b_*\sqcap f)(b_*\sqcap f)^k(\top) = a(b_*(b_*\sqcap f))^k(\top) \sqcap a(f) \sqsubseteq b_*(a(b_*\sqcap f))^k(\top) \sqcap f = (b_*\sqcap f)(a(b_*\sqcap f))^k(\top)$. By induction hypothesis, the latter is equal to $(b_*\sqcap f)(b_*\sqcap f)^k(\top)=(b_*\sqcap f)^{k+1}(\top)$.
\end{proof}

\subsection{Proofs of Section~\ref{sec:companion}}

\begin{proof}[Proof of Lemma~\ref{lemma:inclusion}]
If $a_1\sqsubseteq a_2$, then $a_2(x)\sqsubseteq_C x$ entails that $a_1(x)\sqsubseteq x$. 
Vice versa, assume that $Pre(a_2) \subseteq Pre(a_1)$. Since $a_2(x)\in Pre(a_2)$, it also holds that $a_2(x)\in Pre(a_1)$, i.e., $a_1(a_2(x))\sqsubseteq a_2(x)$. We conclude by using the property that $a_2$ is a closure: $a_1(x) \sqsubseteq a_1(a_2(x))\sqsubseteq a_2(x)$.
\end{proof}

\begin{proof}[Proof of Theorem~\ref{thmcompanion}]
First of all observe that by definition $\Omega_b= Pre(\omega_b)$ and that $\Omega_b$ is closed \wrt{}  $b$, that is $b$ restricts and corestricts to $\Omega_b$. By point $\bullet$3 of Lemma~\ref{lemma:equivalentformulation}, $\omega_b$ is $b$-compatible.
Assume now that $a$ is compatible. Then, again by point $\bullet$3 of Lemma~\ref{lemma:equivalentformulation}, $Pre(a)$ should be closed by $b$. Obviously, $\top$ is a pre-fixed point of $a$. Therefore the chain \[\top \sqsupseteq b(\top)\sqsupseteq bb(\top) \sqsupseteq \dots \] should belong to $Pre(a)$. Since $Pre(a)$ is a complete lattice, then also $\nu b =\bigsqcap_{i\in \Nat} b^i(\top)$ belongs to $Pre(a)$. This means that $\Omega_b \subseteq Pre(a)$ and, by Lemma~\ref{lemma:inclusion}, $a\sqsubseteq \omega_b$. Since, by definition $\Omega_b=Pre(\omega_b)$, the last part of the statement follows immediately by Lemma~\ref{lemma:inclusion}.
\end{proof}

\subsection{Proofs of Section~\ref{sec:localcompleteness}}

\begin{lemma}\label{lemma:localmultiple}
	Let $b,a\colon C \to C$ be a map and a closure operator with associated Galois insertion \( (C,\sqsubseteq_C)\galoiS{\alpha}{\beta} (A,\sqsubseteq_A) \).
\begin{enumerate}
\item $\mu (ab) = \mu(aba)$. 
\item $\gamma (\mu \overline{b}^a) = a(\mu (ab))$.
\item For all $f\in Pre(a)$, $\mu \overline{b}^a \sqsubseteq_A f$ if{}f $\mu (ab)\sqsubseteq_C f$.
\end{enumerate}
\end{lemma}
\begin{proof}
\begin{enumerate}
\item This is Lemma 3.3~\cite{GiacobazziRS00}. For reader's convenience, we report its proof.
The inclusion $\mu(ab)\sqsubseteq \mu(aba)$ is obvious. For the other inclusion, observe that $aba(\mu ab)= aba ab (\mu ab)=abab(\mu ab) = \mu (ab)$, i.e.,  $\mu(ab)$ is a fixed point of $aba$. Thus $\mu(aba)\sqsubseteq \mu(ab)$.
\item Observe that $\gamma (\mu \overline{b}^a)= \gamma \comp  \alpha \comp b \comp \gamma (\mu \overline{b}^a) = ab (\gamma (\mu \overline{b}^a))$, so $\mu ab \sqsubseteq_C \gamma (\mu \overline{b}^a)$. Therefore $\alpha(\mu (ab)) \sqsubseteq \mu \overline{b}^a$.
For the other direction,  $\alpha(\mu aba)=\alpha aba(\mu aba)= \alpha \gamma \alpha b \gamma \alpha (\mu aba) = \alpha b \gamma \alpha (\mu aba)$, that is $\alpha (\mu aba)$ is a fixed point of $\alpha b \gamma$. Therefore $\mu \overline{b}^a \sqsubseteq \alpha(\mu aba) = \alpha (\mu (ab))$.
We thus have  $\alpha(\mu (ab)) = \mu \overline{b}^a$ and then, by monotonicity of $\gamma$,  $\gamma (\mu \overline{b}^a) = a(\mu (ab))$.
\item By point 2 and Lemma~\ref{lemma:absouncomplete}.1.
\end{enumerate}
\end{proof}

\begin{proof}[Proof of Proposition~\ref{prop:localimplies}]
It follows immediately by definition of local completeness and Lemma~\ref{lemma:localmultiple}.3.
\end{proof}

\begin{proof}[Proof of Proposition~\ref{prop:completeifflocal}]
For the top-down implication, assume that  $a$ is $b$-complete. By Lemma~\ref{lemma:alternativecompleteness}, $\mu(ab)\sqsubseteq f$ if{}f $a(\mu b) \sqsubseteq f$. By Lemma~\ref{lemma:absouncomplete}.2, $a(\mu b) \sqsubseteq f$ if{}f $\mu b \sqsubseteq f$, for all $f\in Pre(a)$.

 For bottom up, one can take $f=a(\mu b)$ since $aa(\mu b)=a(\mu b)$. It holds $\mu b \sqsubseteq a(\mu b)$ and thus $\mu(ab) \sqsubseteq a(\mu b)$. The other inclusion $a(\mu b) \sqsubseteq \mu(ab)$ always holds.
\end{proof}

\begin{proof}[Proof of Proposition~\ref{prop:bfclosed}]
If $a_2$ is $(b,f)$-complete, then (1) $a_2(f)\sqsubseteq f$ and (2) $ \mu (a_2 b) \sqsubseteq f$ if{}f  $\mu b \sqsubseteq f$. By (1), we have that $a_1(f)\sqsubseteq a_2(f)\sqsubseteq f$.
Now, as usual, if $\mu(a_1b)\sqsubseteq f$, then  $\mu b\sqsubseteq \mu(a_1b)\sqsubseteq f$. If $\mu b \sqsubseteq f$ then, by (2), $ \mu (a_2 b) \sqsubseteq f$ and thus  $\mu (a_1 b) \sqsubseteq  \mu (a_2 b) \sqsubseteq f$.
\end{proof}

\begin{proof}[Proof of Proposition~\ref{prop:fcompanion}]

\noindent\begin{enumerate}
\item Follows from the fact that the least upper bound of a family of $b$-compatible functions is compatible (see~\cite{pous:aplas07:clut} or Proposition 6.3.11~\cite{PS12}).
\item $\omega_{b,f}(f)=\bigsqcup \{a(f) | a\text{ is $(b,f)$-compatible}\}$. Since $a(f)\sqsubseteq f$ for all $a$ $(b,f)$-compatible, then $\omega_{b,f}(f)\sqsubseteq f$.
\item It is enough to observe that $id$ is $(b,f)$-compatible: $id(f)\sqsubseteq f$ and $id \comp b \sqsubseteq b \comp id$ (Proposition~\ref{prop:mod}.1).
\item By 1 and 2, $\omega_{b,f}$ is $(b,f)$-compatible. It is enough to prove that the composition $a_1\comp a_2$ of two $(b,f)$-compatible maps $a_1,a_2$ is $(b,f)$-compatible. First, $a_1(a_2(f)) \sqsubseteq a_1(f) \sqsubseteq f$. The other property follows from (Proposition~\ref{prop:mod}.2).
\end{enumerate}
\end{proof}

\begin{proof}[Proof of Theorem~\ref{them:localcomp}]
By Proposition~\ref{prop:fcompanion}, $\omega_{b_*,f}$ is $(b_*,f)$-compatible and thus, by Corollary~\ref{cor:bfcomptivility}, $(i\sqcup b^*, f)$-complete. 
Since  $a\sqsubseteq \omega_{b,f}$, by Proposition~\ref{prop:bfclosed}, $a$ is $(i\sqcup b^*, f)$-complete.
\end{proof}

\begin{proof}[Proof of Corollary~\ref{corSBRA}]
Since $b_*$ is a right adjoint, it is cocontinuous, i.e., $b_* \bigsqcap =\bigsqcap b_*$. Therefore the smallest complete lattice containing $f$ and closed \wrt{}  to $b_*$, is
 the complete lattice generated by~\eqref{eq:Omegabf}. If all these generators belong to $Pre(a)$, since $Pre(a)$ is a complete lattice, it holds that $\Omega_{b_*,f}\sqsubseteq Pre(a)$. Then $a\sqsubseteq \omega_{b_*,f}$ and, by Theorem~\ref{them:localcomp}, $a$ is $(i\sqcup b^*, f)$-complete.
\end{proof}

\begin{proof}[Proof of Lemma~\ref{lemma:coinfcomp}]
For monotonicity, take $a_1 \sqsubseteq a_2$. Take $c\colon C \to C$ such that $c b\sqsubseteq ba_1$ and $c(f)\sqsubseteq f$. Then it clearly holds that $c b\sqsubseteq ba_2$. Thus $B(a_1) \sqsubseteq B(a_2)$.

Let us prove now the second property. The right-to-left implication is trivial. For the left-to-right, we have as hypothesis that $a'(x)\sqsubseteq B(a)(x)$ for all $x\in C$. This holds in particular for $b(x)$ and $f$. Therefore $a'(b(x))\sqsubseteq B(a)(b(x))=\bigsqcup \{ cb(x) \mid c b\sqsubseteq ba, c(f)\sqsubseteq f \} \sqsubseteq ba(x)$. Moreover $a'(f)\sqsubseteq B(a)(f) =\bigsqcup \{ cf \mid c b\sqsubseteq ba, c(f)\sqsubseteq f \} \sqsubseteq  f$.
\end{proof}


\begin{thebibliography}{38}

%
%
%
%
%
%
%
%
%
%
%
%
%
%
%
%

\ifx \showCODEN    \undefined \def \showCODEN     #1{\unskip}     \fi
\ifx \showDOI      \undefined \def \showDOI       #1{#1}\fi
\ifx \showISBNx    \undefined \def \showISBNx     #1{\unskip}     \fi
\ifx \showISBNxiii \undefined \def \showISBNxiii  #1{\unskip}     \fi
\ifx \showISSN     \undefined \def \showISSN      #1{\unskip}     \fi
\ifx \showLCCN     \undefined \def \showLCCN      #1{\unskip}     \fi
\ifx \shownote     \undefined \def \shownote      #1{#1}          \fi
\ifx \showarticletitle \undefined \def \showarticletitle #1{#1}   \fi
\ifx \showURL      \undefined \def \showURL       {\relax}        \fi
%
%
\providecommand\bibfield[2]{#2}
\providecommand\bibinfo[2]{#2}
\providecommand\natexlab[1]{#1}
\providecommand\showeprint[2][]{arXiv:#2}

\bibitem[\protect\citeauthoryear{Aho and Hopcroft}{Aho and Hopcroft}{1974}]%
        {Aho:1974:DAC:578775}
\bibfield{author}{\bibinfo{person}{Alfred~V. Aho} {and}
  \bibinfo{person}{John~E. Hopcroft}.} \bibinfo{year}{1974}\natexlab{}.
\newblock \bibinfo{booktitle}{\emph{The Design and Analysis of Computer
  Algorithms} (\bibinfo{edition}{1st} ed.)}.
\newblock \bibinfo{publisher}{Addison-Wesley Longman Publishing Co., Inc.}
\newblock
\showISBNx{0201000296}


\bibitem[\protect\citeauthoryear{Bertot}{Bertot}{2009}]%
        {bertot2009structural}
\bibfield{author}{\bibinfo{person}{Yves Bertot}.}
  \bibinfo{year}{2009}\natexlab{}.
\newblock \showarticletitle{Structural abstract interpretation: A formal study
  using Coq}.
\newblock In \bibinfo{booktitle}{\emph{Language Engineering and Rigorous
  Software Development}}. \bibinfo{series}{LNCS}, Vol.~\bibinfo{volume}{5520}.
  \bibinfo{publisher}{Springer}, \bibinfo{pages}{153--194}.
\newblock
\urldef\tempurl%
\url{https://doi.org/10.1007/978-3-642-03153-3_4}
\showDOI{\tempurl}


\bibitem[\protect\citeauthoryear{Bonchi and Pous}{Bonchi and Pous}{2013}]%
        {bp:popl13:hkc}
\bibfield{author}{\bibinfo{person}{Filippo Bonchi} {and}
  \bibinfo{person}{Damien Pous}.} \bibinfo{year}{2013}\natexlab{}.
\newblock \showarticletitle{Checking {NFA} equivalence with bisimu\-la\-tions
  up to congruence}. In \bibinfo{booktitle}{\emph{{POPL}: 40th ACM
  SIGACT-SIGPLAN symposium on Principles of programming languages}}.
  \bibinfo{publisher}{ACM}, \bibinfo{pages}{457--468}.
\newblock
\showISBNx{978-1-4503-1832-7}


\bibitem[\protect\citeauthoryear{Caucal}{Caucal}{1990}]%
        {Caucal90}
\bibfield{author}{\bibinfo{person}{Didier Caucal}.}
  \bibinfo{year}{1990}\natexlab{}.
\newblock \showarticletitle{Graphes canoniques de graphes alg{\'e}briques}.
\newblock \bibinfo{journal}{\emph{ITA}} \bibinfo{volume}{24},
  \bibinfo{number}{4} (\bibinfo{year}{1990}), \bibinfo{pages}{339--352}.
\newblock
\urldef\tempurl%
\url{http://www.numdam.org/item?id=ITA_1990__24_4_339_0}
\showURL{%
\tempurl}


\bibitem[\protect\citeauthoryear{Cousot}{Cousot}{1997}]%
        {Cousot1997}
\bibfield{author}{\bibinfo{person}{Patrick Cousot}.}
  \bibinfo{year}{1997}\natexlab{}.
\newblock \showarticletitle{Types as abstract interpretations}. In
  \bibinfo{booktitle}{\emph{{POPL}: 24th ACM SIGACT-SIGPLAN symposium on
  Principles of programming languages}}. \bibinfo{publisher}{{ACM}},
  \bibinfo{pages}{316--331}.
\newblock
\urldef\tempurl%
\url{https://doi.org/10.1145/263699.263744}
\showDOI{\tempurl}


\bibitem[\protect\citeauthoryear{Cousot}{Cousot}{2000}]%
        {Cousot2000}
\bibfield{author}{\bibinfo{person}{Patrick Cousot}.}
  \bibinfo{year}{2000}\natexlab{}.
\newblock \showarticletitle{Partial Completeness of Abstract Fixpoint
  Checking}.
\newblock In \bibinfo{booktitle}{\emph{{SARA}: Int. Symp. on Abstraction,
  Reformulation, and Approximation}}. \bibinfo{publisher}{Springer},
  \bibinfo{pages}{1--25}.
\newblock
\urldef\tempurl%
\url{https://doi.org/10.1007/3-540-44914-0_1}
\showDOI{\tempurl}


\bibitem[\protect\citeauthoryear{Cousot and Cousot}{Cousot and Cousot}{1977}]%
        {cousot1977abstract}
\bibfield{author}{\bibinfo{person}{Patrick Cousot} {and}
  \bibinfo{person}{Radhia Cousot}.} \bibinfo{year}{1977}\natexlab{}.
\newblock \showarticletitle{Abstract interpretation: a unified lattice model
  for static analysis of programs by construction or approximation of
  fixpoints}. In \bibinfo{booktitle}{\emph{{POPL}: 4th ACM SIGACT-SIGPLAN
  symposium on Principles of programming languages}}. \bibinfo{publisher}{ACM},
  \bibinfo{pages}{238--252}.
\newblock


\bibitem[\protect\citeauthoryear{Cousot and Cousot}{Cousot and Cousot}{1979a}]%
        {CC79b}
\bibfield{author}{\bibinfo{person}{Patrick Cousot} {and}
  \bibinfo{person}{Radhia Cousot}.} \bibinfo{year}{1979}\natexlab{a}.
\newblock \showarticletitle{A constructive characterization of the lattices of
  all retractions, preclosure, quasi-closure and closure operators on a
  complete lattice}.
\newblock \bibinfo{journal}{\emph{Portug.\ Math.}} \bibinfo{volume}{38},
  \bibinfo{number}{2} (\bibinfo{year}{1979}), \bibinfo{pages}{185--198}.
\newblock


\bibitem[\protect\citeauthoryear{Cousot and Cousot}{Cousot and Cousot}{1979b}]%
        {cousot1979systematic}
\bibfield{author}{\bibinfo{person}{Patrick Cousot} {and}
  \bibinfo{person}{Radhia Cousot}.} \bibinfo{year}{1979}\natexlab{b}.
\newblock \showarticletitle{Systematic design of program analysis frameworks}.
  In \bibinfo{booktitle}{\emph{{POPL}: 6th ACM SIGACT-SIGPLAN symposium on
  Principles of programming languages}}. \bibinfo{publisher}{ACM},
  \bibinfo{pages}{269--282}.
\newblock


\bibitem[\protect\citeauthoryear{Cousot and Cousot}{Cousot and Cousot}{1999}]%
        {CousotC99}
\bibfield{author}{\bibinfo{person}{Patrick Cousot} {and}
  \bibinfo{person}{Radhia Cousot}.} \bibinfo{year}{1999}\natexlab{}.
\newblock \showarticletitle{Refining Model Checking by Abstract
  Interpretation}.
\newblock \bibinfo{journal}{\emph{Automated Software Engineering}}
  \bibinfo{volume}{6}, \bibinfo{number}{1} (\bibinfo{year}{1999}),
  \bibinfo{pages}{69--95}.
\newblock


\bibitem[\protect\citeauthoryear{Cousot, Cousot, Feret, Mauborgne, Min\'e,
  Monniaux, and Rival}{Cousot et~al\mbox{.}}{2005}]%
        {Esop05Astree}
\bibfield{author}{\bibinfo{person}{Patrick Cousot}, \bibinfo{person}{Radhia
  Cousot}, \bibinfo{person}{J{\'e}r{\^o}me Feret}, \bibinfo{person}{Laurent
  Mauborgne}, \bibinfo{person}{Antoine Min\'e}, \bibinfo{person}{David
  Monniaux}, {and} \bibinfo{person}{Xavier Rival}.}
  \bibinfo{year}{2005}\natexlab{}.
\newblock \showarticletitle{The {ASTR\'EE} Analyzer}. In
  \bibinfo{booktitle}{\emph{{ESOP}: European Symposium on Programming}}
  \emph{(\bibinfo{series}{LNCS})}, Vol.~\bibinfo{volume}{3444}.
  \bibinfo{publisher}{Springer}, \bibinfo{pages}{21--30}.
\newblock


\bibitem[\protect\citeauthoryear{Cousot, Ganty, and Raskin}{Cousot
  et~al\mbox{.}}{2007}]%
        {CGR07}
\bibfield{author}{\bibinfo{person}{Patrick Cousot}, \bibinfo{person}{Pierre
  Ganty}, {and} \bibinfo{person}{Jean-Fran\c{c}ois Raskin}.}
  \bibinfo{year}{2007}\natexlab{}.
\newblock \showarticletitle{Fixpoint-Guided Abstraction Refinements}. In
  \bibinfo{booktitle}{\emph{{SAS}: 14th Int. Static Analysis Symp.}}
  \emph{(\bibinfo{series}{LNCS})}, Vol.~\bibinfo{volume}{4634}.
  \bibinfo{publisher}{Springer}, \bibinfo{pages}{333--348}.
\newblock
\urldef\tempurl%
\url{https://doi.org/10.1007/978-3-540-74061-2_21}
\showDOI{\tempurl}


\bibitem[\protect\citeauthoryear{Danielsson}{Danielsson}{2017}]%
        {danielsson2017up}
\bibfield{author}{\bibinfo{person}{Nils~Anders Danielsson}.}
  \bibinfo{year}{2017}\natexlab{}.
\newblock \showarticletitle{Up-to Techniques Using Sized Types}.
\newblock \bibinfo{journal}{\emph{Proc. ACM Program. Lang.}}
  \bibinfo{volume}{2}, \bibinfo{number}{POPL} (\bibinfo{year}{2017}),
  \bibinfo{pages}{43:1--43:28}.
\newblock
\showISSN{2475-1421}
\urldef\tempurl%
\url{https://doi.org/10.1145/3158131}
\showDOI{\tempurl}


\bibitem[\protect\citeauthoryear{Davey and Priestley}{Davey and
  Priestley}{2002}]%
        {davey_priestley_2002}
\bibfield{author}{\bibinfo{person}{B.~A. Davey} {and} \bibinfo{person}{H.~A.
  Priestley}.} \bibinfo{year}{2002}\natexlab{}.
\newblock \bibinfo{booktitle}{\emph{Introduction to Lattices and Order}
  (\bibinfo{edition}{2} ed.)}.
\newblock \bibinfo{publisher}{Cambridge University Press}.
\newblock
\urldef\tempurl%
\url{https://doi.org/10.1017/CBO9780511809088}
\showDOI{\tempurl}


\bibitem[\protect\citeauthoryear{F\"{a}hndrich and Logozzo}{F\"{a}hndrich and
  Logozzo}{2011}]%
        {FahndrichL2010}
\bibfield{author}{\bibinfo{person}{Manuel F\"{a}hndrich} {and}
  \bibinfo{person}{Francesco Logozzo}.} \bibinfo{year}{2011}\natexlab{}.
\newblock \showarticletitle{Static Contract Checking with Abstract
  Interpretation}. In \bibinfo{booktitle}{\emph{Int. Conf. on Formal
  Verification of Object-oriented Software}}
  \emph{(\bibinfo{series}{FoVeOOS'10})}. \bibinfo{publisher}{Springer-Verlag},
  \bibinfo{address}{Berlin, Heidelberg}, \bibinfo{pages}{10--30}.
\newblock
\showISBNx{3-642-18069-8, 978-3-642-18069-9}


\bibitem[\protect\citeauthoryear{Giacobazzi, Logozzo, and Ranzato}{Giacobazzi
  et~al\mbox{.}}{2015}]%
        {GiacobazziLR15}
\bibfield{author}{\bibinfo{person}{Roberto Giacobazzi},
  \bibinfo{person}{Francesco Logozzo}, {and} \bibinfo{person}{Francesco
  Ranzato}.} \bibinfo{year}{2015}\natexlab{}.
\newblock \showarticletitle{Analyzing Program Analyses}. In
  \bibinfo{booktitle}{\emph{{POPL}: 42nd ACM SIGACT-SIGPLAN symposium on
  Principles of programming languages}}. \bibinfo{publisher}{{ACM}},
  \bibinfo{pages}{261--273}.
\newblock
\showISBNx{978-1-4503-3300-9}


\bibitem[\protect\citeauthoryear{Giacobazzi and Quintarelli}{Giacobazzi and
  Quintarelli}{2001}]%
        {GiacobazziQ01}
\bibfield{author}{\bibinfo{person}{Roberto Giacobazzi} {and}
  \bibinfo{person}{Elisa Quintarelli}.} \bibinfo{year}{2001}\natexlab{}.
\newblock \showarticletitle{Incompleteness, Counterexamples, and Refinements in
  Abstract Model-Checking}. In \bibinfo{booktitle}{\emph{{SAS}: 8th Int. Static
  Analysis Symposium}} \emph{(\bibinfo{series}{LNCS})},
  Vol.~\bibinfo{volume}{2126}. \bibinfo{publisher}{Springer},
  \bibinfo{pages}{356--373}.
\newblock


\bibitem[\protect\citeauthoryear{Giacobazzi, Ranzato, and Scozzari}{Giacobazzi
  et~al\mbox{.}}{2000}]%
        {GiacobazziRS00}
\bibfield{author}{\bibinfo{person}{Roberto Giacobazzi},
  \bibinfo{person}{Francesco Ranzato}, {and} \bibinfo{person}{Francesca
  Scozzari}.} \bibinfo{year}{2000}\natexlab{}.
\newblock \showarticletitle{Making abstract interpretations complete}.
\newblock \bibinfo{journal}{\emph{J. {ACM}}} \bibinfo{volume}{47},
  \bibinfo{number}{2} (\bibinfo{year}{2000}), \bibinfo{pages}{361--416}.
\newblock


\bibitem[\protect\citeauthoryear{Hopcroft}{Hopcroft}{1971}]%
        {hopcroft1971n}
\bibfield{author}{\bibinfo{person}{John Hopcroft}.}
  \bibinfo{year}{1971}\natexlab{}.
\newblock \bibinfo{booktitle}{\emph{An N Log N Algorithm for Minimizing States
  in a Finite Automaton}}.
\newblock \bibinfo{type}{{T}echnical {R}eport}. \bibinfo{institution}{Stanford
  Univ Calif Dept of Computer Science}.
\newblock


\bibitem[\protect\citeauthoryear{Hopcroft and Karp}{Hopcroft and Karp}{1971}]%
        {HopcroftKarp}
\bibfield{author}{\bibinfo{person}{John~E. Hopcroft} {and}
  \bibinfo{person}{Richard~M. Karp}.} \bibinfo{year}{1971}\natexlab{}.
\newblock \bibinfo{booktitle}{\emph{A Linear Algorithm for Testing Equivalence
  of Finite Automata}}.
\newblock \bibinfo{type}{{T}echnical {R}eport} 114.
  \bibinfo{institution}{Cornell Univ.}
\newblock


\bibitem[\protect\citeauthoryear{Hur, Neis, Dreyer, and Vafeiadis}{Hur
  et~al\mbox{.}}{2013}]%
        {hur2013power}
\bibfield{author}{\bibinfo{person}{Chung-Kil Hur}, \bibinfo{person}{Georg
  Neis}, \bibinfo{person}{Derek Dreyer}, {and} \bibinfo{person}{Viktor
  Vafeiadis}.} \bibinfo{year}{2013}\natexlab{}.
\newblock \showarticletitle{The Power of Parameterization in Coinductive
  Proof}. In \bibinfo{booktitle}{\emph{{POPL}: 40th ACM SIGPLAN-SIGACT
  Symposium on Principles of Programming Languages}}. \bibinfo{publisher}{ACM},
  \bibinfo{pages}{193--206}.
\newblock


\bibitem[\protect\citeauthoryear{Klin}{Klin}{2011}]%
        {DBLP:journals/tcs/Klin11}
\bibfield{author}{\bibinfo{person}{Bartek Klin}.}
  \bibinfo{year}{2011}\natexlab{}.
\newblock \showarticletitle{Bialgebras for structural operational semantics: An
  introduction}.
\newblock \bibinfo{journal}{\emph{Theoretical Computer Science}}
  \bibinfo{volume}{412}, \bibinfo{number}{38} (\bibinfo{year}{2011}),
  \bibinfo{pages}{5043--5069}.
\newblock


\bibitem[\protect\citeauthoryear{Leroy}{Leroy}{2014}]%
        {leroy2014formal}
\bibfield{author}{\bibinfo{person}{Xavier Leroy}.}
  \bibinfo{year}{2014}\natexlab{}.
\newblock \showarticletitle{Formal verification of a static analyzer: abstract
  interpretation in type theory}. In \bibinfo{booktitle}{\emph{Types-The 2014
  Types Meeting}}.
\newblock


\bibitem[\protect\citeauthoryear{Milner}{Milner}{1989}]%
        {Milner89}
\bibfield{author}{\bibinfo{person}{Robert Milner}.}
  \bibinfo{year}{1989}\natexlab{}.
\newblock \bibinfo{booktitle}{\emph{Communication and Concurrency}}.
\newblock \bibinfo{publisher}{Prentice\,Hall}.
\newblock


\bibitem[\protect\citeauthoryear{O{\textquotesingle}Hearn}{O{\textquotesingle}Hearn}{2015}]%
        {OHearn2015}
\bibfield{author}{\bibinfo{person}{Peter O{\textquotesingle}Hearn}.}
  \bibinfo{year}{2015}\natexlab{}.
\newblock \showarticletitle{From Categorical Logic to Facebook Engineering}. In
  \bibinfo{booktitle}{\emph{{LICS}: 30th Annual {ACM}/{IEEE} Symposium on Logic
  in Computer Science}}. \bibinfo{publisher}{{IEEE}}, \bibinfo{pages}{17--20}.
\newblock


\bibitem[\protect\citeauthoryear{Park}{Park}{1969}]%
        {Park69}
\bibfield{author}{\bibinfo{person}{David~M.R. Park}.}
  \bibinfo{year}{1969}\natexlab{}.
\newblock \showarticletitle{Fixpoint induction and proofs of program
  properties}. In \bibinfo{booktitle}{\emph{Machine Intelligence}},
  Vol.~\bibinfo{volume}{5}. \bibinfo{publisher}{Edinburgh Univ. Press},
  \bibinfo{pages}{59--78}.
\newblock


\bibitem[\protect\citeauthoryear{Parrow and Weber}{Parrow and Weber}{2016}]%
        {parrow2016largest}
\bibfield{author}{\bibinfo{person}{Joachim Parrow} {and} \bibinfo{person}{Tjark
  Weber}.} \bibinfo{year}{2016}\natexlab{}.
\newblock \showarticletitle{The largest respectful function}.
\newblock \bibinfo{journal}{\emph{arXiv preprint arXiv:1605.04136}}
  (\bibinfo{year}{2016}).
\newblock


\bibitem[\protect\citeauthoryear{Pous}{Pous}{2007}]%
        {pous:aplas07:clut}
\bibfield{author}{\bibinfo{person}{Damien Pous}.}
  \bibinfo{year}{2007}\natexlab{}.
\newblock \showarticletitle{Complete Lattices and Up-To Techniques}. In
  \bibinfo{booktitle}{\emph{{APLAS}: Asian Symposium on Programming Languages
  and Systems}} \emph{(\bibinfo{series}{LNCS})}, Vol.~\bibinfo{volume}{4807}.
  \bibinfo{publisher}{Springer}, \bibinfo{pages}{351--366}.
\newblock
\urldef\tempurl%
\url{https://doi.org/10.1007/978-3-540-76637-7_24}
\showDOI{\tempurl}


\bibitem[\protect\citeauthoryear{Pous}{Pous}{2016a}]%
        {pous2016coinduction}
\bibfield{author}{\bibinfo{person}{Damien Pous}.}
  \bibinfo{year}{2016}\natexlab{a}.
\newblock \showarticletitle{Coinduction all the way up}. In
  \bibinfo{booktitle}{\emph{{LICS}: 31st Annual ACM/IEEE Symposium on Logic in
  Computer Science}}. \bibinfo{publisher}{ACM}, \bibinfo{pages}{307--316}.
\newblock


\bibitem[\protect\citeauthoryear{Pous}{Pous}{2016b}]%
        {damienlibrary}
\bibfield{author}{\bibinfo{person}{Damien Pous}.}
  \bibinfo{year}{2016}\natexlab{b}.
\newblock \bibinfo{title}{{Coq libraries for ``Coinduction all the way up''}}.
\newblock
  \bibinfo{howpublished}{\url{http://perso.ens-lyon.fr/damien.pous/cawu/}}.
  (\bibinfo{year}{2016}).
\newblock


\bibitem[\protect\citeauthoryear{Pous and Sangiorgi}{Pous and
  Sangiorgi}{2012}]%
        {PS12}
\bibfield{author}{\bibinfo{person}{Damien Pous} {and} \bibinfo{person}{Davide
  Sangiorgi}.} \bibinfo{year}{2012}\natexlab{}.
\newblock \showarticletitle{Enhancements of the bisimulation proof method}.
\newblock In \bibinfo{booktitle}{\emph{Advanced Topics in Bisimulation and
  Coinduction}}. \bibinfo{publisher}{Cambridge University Press},
  \bibinfo{pages}{233--289}.
\newblock
\showISBNx{9781107004979}


\bibitem[\protect\citeauthoryear{Ranzato, Doria, and Tapparo}{Ranzato
  et~al\mbox{.}}{2008}]%
        {RanzatoRT08}
\bibfield{author}{\bibinfo{person}{Francesco Ranzato},
  \bibinfo{person}{Olivia~Rossi Doria}, {and} \bibinfo{person}{Francesco
  Tapparo}.} \bibinfo{year}{2008}\natexlab{}.
\newblock \showarticletitle{A Forward-Backward Abstraction Refinement
  Algorithm}.
\newblock In \bibinfo{booktitle}{\emph{{VMCAI}: 9th Int. Conf. on Verification,
  Model Checking, and Abstract Interpretation}}. \bibinfo{series}{LNCS},
  Vol.~\bibinfo{volume}{4905}. \bibinfo{publisher}{Springer},
  \bibinfo{pages}{248--262}.
\newblock
\urldef\tempurl%
\url{https://doi.org/10.1007/978-3-540-78163-9_22}
\showDOI{\tempurl}


\bibitem[\protect\citeauthoryear{Ranzato and Tapparo}{Ranzato and
  Tapparo}{2007a}]%
        {RanzatoT07}
\bibfield{author}{\bibinfo{person}{Francesco Ranzato} {and}
  \bibinfo{person}{Francesco Tapparo}.} \bibinfo{year}{2007}\natexlab{a}.
\newblock \showarticletitle{Generalized Strong Preservation by Abstract
  Interpretation}.
\newblock \bibinfo{journal}{\emph{J. Log. Comput.}} \bibinfo{volume}{17},
  \bibinfo{number}{1} (\bibinfo{year}{2007}), \bibinfo{pages}{157--197}.
\newblock
\urldef\tempurl%
\url{https://doi.org/10.1093/logcom/exl035}
\showDOI{\tempurl}


\bibitem[\protect\citeauthoryear{Ranzato and Tapparo}{Ranzato and
  Tapparo}{2007b}]%
        {RanzatoTLics07}
\bibfield{author}{\bibinfo{person}{Francesco Ranzato} {and}
  \bibinfo{person}{Francesco Tapparo}.} \bibinfo{year}{2007}\natexlab{b}.
\newblock \showarticletitle{A New Efficient Simulation Equivalence Algorithm}.
  In \bibinfo{booktitle}{\emph{{LICS}: 22nd Annual {IEEE} Symposium on Logic in
  Computer Science}}. \bibinfo{publisher}{{IEEE}}, \bibinfo{pages}{171--180}.
\newblock
\urldef\tempurl%
\url{https://doi.org/10.1109/lics.2007.8}
\showDOI{\tempurl}


\bibitem[\protect\citeauthoryear{Sangiorgi}{Sangiorgi}{1998}]%
        {San98MFCS}
\bibfield{author}{\bibinfo{person}{Davide Sangiorgi}.}
  \bibinfo{year}{1998}\natexlab{}.
\newblock \showarticletitle{On the Bisimulation Proof Method}.
\newblock \bibinfo{journal}{\emph{Mathematical Structures in Computer Science}}
   \bibinfo{volume}{8} (\bibinfo{year}{1998}), \bibinfo{pages}{447--479}.
\newblock


\bibitem[\protect\citeauthoryear{Turi and Plotkin}{Turi and Plotkin}{1997}]%
        {turi1997towards}
\bibfield{author}{\bibinfo{person}{Daniele Turi} {and} \bibinfo{person}{Gordon
  Plotkin}.} \bibinfo{year}{1997}\natexlab{}.
\newblock \showarticletitle{Towards a mathematical operational semantics}. In
  \bibinfo{booktitle}{\emph{{LICS}: 12th Annual IEEE Symposium on Logic in
  Computer Science}}. \bibinfo{publisher}{IEEE}, \bibinfo{pages}{280--291}.
\newblock
\urldef\tempurl%
\url{https://doi.org/10.1109/LICS.1997.614955}
\showDOI{\tempurl}


\bibitem[\protect\citeauthoryear{VERASCO}{VERASCO}{2015}]%
        {verasco}
\bibfield{author}{\bibinfo{person}{VERASCO}.} \bibinfo{year}{2015}\natexlab{}.
\newblock \bibinfo{title}{A Formally-Verified Static Analyzer for {C}}.
\newblock \bibinfo{howpublished}{\url{http://compcert.inria.fr/verasco/}}.
  (\bibinfo{year}{2015}).
\newblock


\bibitem[\protect\citeauthoryear{Ward}{Ward}{1942}]%
        {Ward42}
\bibfield{author}{\bibinfo{person}{M. Ward}.} \bibinfo{year}{1942}\natexlab{}.
\newblock \showarticletitle{The closure operators of a lattice}.
\newblock \bibinfo{journal}{\emph{Ann.\ Math.}} \bibinfo{volume}{43},
  \bibinfo{number}{2} (\bibinfo{year}{1942}), \bibinfo{pages}{191--196}.
\newblock


\end{thebibliography}
\end{document}